\documentclass[12pt, a4paper,onecolumn]{belarticle2}
\usepackage{amsmath, amsfonts, amssymb, enumerate, fullpage, amsthm, graphicx, braket, relsize, bbm, mathrsfs, tikz}
\usepackage{graphicx}
\usetikzlibrary{calc}
\usetikzlibrary{shapes,decorations}
\usepackage{geometry}
\usepackage{subfigure}
\newcommand{\tr}[1]{\mathrm{tr}\left\{ #1 \right\}}
\newcommand{\Tr}[2]{\mathrm{tr}_{#2}\left\{ #1 \right\}}
\newcommand{\id}{\mathbbm{1}}

\usepackage{soul}
\usepackage{ulem}

\usepackage{tikz}
\usetikzlibrary{patterns}
\usetikzlibrary{calc,decorations.pathreplacing}
\usetikzlibrary{arrows,shapes}
\usepackage{subfigure}

\newtheorem{theo}{Theorem}
\newtheorem{thm}[theo]{Theorem}

\newtheorem{lem}[theo]{Lemma}

\newtheorem{defn}[theo]{Definition}
\newtheorem{rem}[theo]{Remark}
\newtheorem{ex}[theo]{Example}

\newcommand{\blam}{\boldsymbol{\lambda}}
\newcommand{\bsig}{\boldsymbol{\sigma}}
\newcommand{\ba}{\mathbf{a}}
\newcommand{\bx}{\mathbf{x}}

\let\OLDthebibliography\thebibliography
\renewcommand\thebibliography[1]{
  \OLDthebibliography{#1}
  \setlength{\parskip}{0pt}
  \setlength{\itemsep}{0pt plus 0.3ex}
}

\begin{document}
\title{A formalism for steering with local quantum \\ measurements}
\author{A.~B.~Sainz}
\affiliation{Perimeter Institute for Theoretical Physics, 31 Caroline St. N, Waterloo, Ontario, Canada, N2L 2Y5.}
\author{L.~Aolita}
\affiliation{Instituto  de  F\'isica,  Universidade  Federal  do  Rio  de  Janeiro, Caixa  Postal  68528,  21941-972  Rio  de  Janeiro,  RJ,  Brazil.}
\author{M.~Piani}
\affiliation{SUPA and Department of Physics, University of Strathclyde, Glasgow G4 0NG, UK.}
\author{M.~J.~Hoban}
\affiliation{University of Oxford, Department of Computer Science, Wolfson Building, Parks Road, OX1 3QD, UK.}
\affiliation{School of Informatics, University of Edinburgh, 10 Crichton Street, Edinburgh EH8 9AB, UK}
\author{P.~Skrzypczyk}
\affiliation{H. H. Wills Physics Laboratory, University of Bristol, Tyndall Avenue, Bristol, BS8 1TL, UK.}

\date{\today}

\maketitle
\begin{abstract}
We develop a unified approach to classical, quantum and post-quantum steering. The framework is based on uncharacterised (black-box) parties performing quantum measurements on their share of a (possibly unphysical) quantum state, and its starting point is the characterisation of general no-signalling assemblages via non-positive local hidden-state models. By developing a connection to entanglement witnesses, this formalism allows for new definitions of families of assemblages, in particular via (i) non-decomposable positive maps and (ii) unextendible product bases. The former proves to be useful for constructing post-quantum assemblages with the built-in feature of yielding only quantum correlations in Bell experiments, while the latter always gives certifiably post-quantum assemblages. Finally, our framework is equipped with an inherent quantifier of post-quantum steering, which we call the negativity of post-quantum steering. We postulate that post-quantum steering should not increase under one-way quantum operations from the steered parties to the steering parties, and we show that, in this sense, the negativity of post-quantum steering is a convex post-quantum-steering monotone.
\end{abstract}


The concept of steering was first introduced by E.~Schr\"odinger in 1935 \cite{st1} in response to the Einstein, Podolsky and Rosen paradox \cite{epr}. It refers to the phenomenon where one party, Alice, by performing measurements on one part of a shared system, seemingly remotely `steers' the state of the system held by a distant party, Bob, in a way which has no explanation in terms of local causal influences. Steering has only recently been formally defined in a quantum information-theoretic setting \cite{st2}, as a way of certifying the entanglement of quantum systems without the need to trust one of the parties, or when one of the parties is using uncharacterised devices. In this setting, the uncharacterised party convinces the other party that they shared entanglement by demonstrating steering. Furthermore, if all parties are uncharacterised (or untrusted) then one recovers the device-independent setting of a standard Bell test. Steering thus may be seen as one in a family of non-classical phenomena, closely related to entanglement and Bell non-locality \cite{Bell}. Indeed, Bell non-locality implies steering, and steering implies entanglement, however all three concepts are inequivalent \cite{st2,quintino2015}.

It is well known that, in spite of demonstrating non-locality, local measurements on entangled quantum systems cannot be used to communicate superluminally. That is, correlations that are  generated by varying the choice of local measurements on space-like separated quantum subsystem -- which we define to be \textit{quantum correlations} -- satisfy the principle of no-signalling. We will call \textit{no-signalling colleations} all correlations that do not permit signalling. One can conceive of no-signalling correlations that cannot be realised by local measurements on quantum states, hence called \textit{post-quantum correlations}; this possibility was first pointed out in a seminal work by Popescu and Rohrlich \cite{PR}. A pertinent question at the heart of quantum foundations since then has regarded the reason why we do not seem to observe these post-quantum correlations in nature \cite{popescu}. This line of questioning has resulted in the proposal of physical and information-theoretic axioms that aim to single out the set of quantum correlations among the no-signalling correlations \cite{ntcc, ntcc2, IC, nanlc, ML, LO, aq}. 

Since Bell non-locality implies steering, it is natural that there should also exist \textit{post-quantum steering},  i.e. steering that does not lead to superluminal signalling yet cannot be realised through local measurements on a quantum system. In the standard steering scenario -- only two parties, one of whom is uncharacterised and the other who holds a quantum system -- there is no such thing as post-quantum steering: the only ways in which a single Alice can steer a quantum Bob without leading to signalling have a quantum explanation \cite{ghjw}. However, it is possible to have post-quantum steering in some multipartite generalisations of steering. Such situations involve three or more parties, with at least two uncharacterised parties, as first pointed out in Ref. \cite{PQP}. Remarkably, it is possible to have post-quantum steering without the presence of post-quantum non-locality, demonstrating that these two concepts are in fact intrinsically distinct \cite{PQP}.

The question of how to best understand post-quantum steering, including its possibilities and its limitations -- which could ultimately lead to an information-theoretic reason why post-quantum steering does not appear in nature -- is still open. One main reason for this is the lack of a framework within which to study quantum as well as post-quantum steering in a unified manner. This makes the implications of post-quantum steering difficult to address. We cannot take a black-box approach -- that is, based solely on the use of conditional probability distributions, as in the case of Bell non-locality -- since there is the assumption that one or more parties have a quantum system and their devices are well-characterised. Nevertheless, in the steering framework there is a natural analogue to conditional probability distributions: the \textit{assemblage}. The latter is the collection of states of the characterised parties for each possible measurement outcome of measurements made by the uncharacterised systems. Another obstacle on the path towards understanding the power of post-quantum steering in information tasks is the lack of examples of (large families of) post-quantum steering assemblages.

In this work we develop a framework for steering based on that of \cite{toni} (see also \cite{barnum}) for Bell non-locality. In this formalism, the parties share a (potentially non-quantum) system in the (potentially unphysical) state $\tilde{\rho}$, where some parties steer the others by performing quantum measurements on their share of the system. By unphysical state we mean that $\tilde{\rho}$ is not necessarily positive semi-definite, but it is Hermitian and has unit trace. We show that different families of assemblages arise naturally within the framework depending on the properties of the operator $\tilde{\rho}$, and in this way we can identify assemblages with a local hidden state model, as well as quantum and general no-signalling assemblages. 

Furthermore, we describe a new family of assemblages, which we call Gleason assemblages, in analogue to Gleason correlations \cite{toni}. These are assemblages that arise when $\tilde{\rho}$ is an \textit{entanglement witnesses}. Motivated by the fact that every positive (but not completely positive) map results in an entanglement witness \cite{horodecki}, we consider a novel means of generating post-quantum assemblages: the application of positive (but not completely positive) maps to the quantum systems held by the characterised parties--equivalently, to the assemblage. We show that this construction automatically leads to quantum correlations upon measuring the characterised systems, yet can lead to post-quantum assemblages when a special class of positive maps is considered (so-called non-decomposable maps). In other words, we present a constructive way of generating post-quantum assemblages that only produce quantum correlations. This provides the first general analyatic construction of post-quantum steering without post-quantum nonlocality, with the only known examples to date being obtained through numerical optimisation \cite{PQP}.

We also study assemblages that arise when the parties perform local measurements on entanglement witnesses constructed from an unextendible product basis \cite{LO2}. This is a simple construction that always yields certifiable post-quantum assemblages (although with post-quantum correlations). In addition, we provide a characterisation of general no-signalling assemblages as affine combinations of local hidden state assemblages. This result, which generalises that of Ref.~\cite{tony} for Bell scenarios, not only provides an operational interpretation for non-classical assemblages but also serves as a useful tool for developing our work further. Finally, our framework also provides an inherent post-quantum steering quantifier in terms of the minimal negativity of the operator $\tilde{\rho}$ necessary to reproduce a given assemblage. We prove that such a quantifier does not increase under processing of the assemblage by means of one-way quantum operations from the steered party to the steering party, whereas standard steering is postulated not to increase under one-way local operations and classical communication.

The outline of the paper is as follows. In Sec.~\ref{sec1} we introduce the concept of steering and local hidden state models. Then in the next two sections (Secs.~\ref{ap:negLHS} and \ref{sec3}) we introduce a generalisation of local hidden state models that can account for general no-signalling assemblages. The tools developed in these sections allow us to introduce our general formalism for steering in quantum theory and beyond in Sec.~\ref{se:for}, and then introduce the notion of Gleason assemblages. The direct connection between entanglement witnesses and positive but not completely positive maps is then exploited in Sec.~\ref{sec5} to generate new examples of post-quantum steering without post-quantum non-locality. In Sec.~\ref{sec:UPB} we generate post-quantum assemblages using entanglement witnesses constructed from unextendible product bases. In Sec.~\ref{sec:quantifier} we introduce a quantifier of post-quantum steering, proving its monotonicity under one-way quantum operations. We conclude with some remarks and open problems.

\section{Steering}\label{sec1}

Let us start by describing the simplest steering scenario consisting of two separated parties, Alice and Bob. The roles these parties play in the experiment are different: 
Alice (a.k.a. the `steering' party) is thought of as having a black box, where she decides on an input $x$ and obtains an outcome $a$. Nothing is assumed about the inner workings of this device. On the other hand, the situation at Bob's lab (who is known as the `steered' party) is fully described by means of quantum mechanics: he has access to a system whose marginal state is given by $\rho_R$. Each round in the experiment consists of Alice choosing an input $x$ and obtaining an outcome $a$, with probability $p(a|x)$, and Bob obtaining the conditional marginal state $\rho_{a|x}$ into which his system has been steered. It is convenient to work with the unnormalised steered states $\sigma_{a|x} := p(a|x)\rho_{a|x}$ which contains information both about the probabilities of the steering party, $p(a|x) = \tr{\sigma_{a|x}}$, and the conditional marginal states $\rho_{a|x} = \sigma_{a|x}/p(a|x)$.

The first relevant question in such a set-up is: given a set $\bsig_{A\vert X}:=\{\sigma_{a|x}\}_{a,x}$ of conditional states $\sigma_{a|x}$, which we shall refer to as an \textit{assemblage}, prepared in Bob's lab, could it have arisen by Alice and Bob performing measurements on a classically correlated shared system?

In general, in a quantum scenario, the elements of the assemblage are given by
\begin{equation}
\label{eq:quantumassemblage}
\sigma_{a|x} = \Tr{(M_{a|x} \otimes \id )\, \rho}{A},
\end{equation}
where $\rho$ is a state shared by Alice and Bob, and $M_{a|x}$ is the $a$-th element of a general measurement on Alice's subsystem -- i.e., a positive-operator valued measure (POVM) -- $M_x:=\{M_{a|x}\}_{a,x}$, with $M_{a|x}\geq 0$ and $\sum_{a} M_{a|x}=\openone$.

A separable (or classically correlated) bipartite state has the structure
\begin{equation}
\label{eq:separable}
\rho=\sum_\lambda p_\lambda \rho^A_{\lambda}\otimes \rho^B_\lambda,
\end{equation}
with $\{p_\lambda\}$ a probability distribution, and each $\rho^A_{\lambda}$ a normalized state for $A$ (similarly for $B$).

If $\rho$ in (\ref{eq:quantumassemblage}) can be chosen to be separable, that is, as in (\ref{eq:separable}), the experiment is said to have a \textit{local hidden state} (LHS) model, and the members of the assemblage can be written as
\begin{align}
\sigma_{a|x} = \sum_\lambda p_\lambda(a|x) \, \sigma_\lambda\,,
\end{align}
where $\sigma_{\lambda} \geq 0$ are sub-normalised quantum states such that $p(\lambda):=\tr{\sigma_\lambda}$ satisfies $\sum_\lambda p(\lambda) = 1$, and $p_\lambda(a|x)$ are well-defined conditional probability distributions for all $\lambda$. With respect to the notation of
\eqref{eq:separable}, one would have $p_{\lambda}(a|x)=\Tr{M_{a|x}\rho_{\lambda}^A}{}$ and $\sigma_\lambda=\rho_\lambda^B/p_\lambda$.

 Conversely, whenever the conditional states $\sigma_{a|x}$ do not admit an LHS model -- that is, they cannot arise from local measurements on a separable state -- it is said that steering has been demonstrated from Alice to Bob, and in this case, a state $\rho$ that is entangled is necessarily shared between Alice and Bob in order to satisfy (\ref{eq:quantumassemblage}).

In the literature, the steering (resp.~steered) party is also sometimes said to be uncharacterised (resp.~characterised) or untrusted (resp.~trusted), depending on the particular context in which the steering experiment is performed (for instance, a cryptographic scenario). In this manuscript, we will use these names synonymously without inheriting any of their implicit assumptions on the nature or circumstances of the set-up. 

We are also interested in situations beyond the standard bipartite steering scenario, involving an arbitrary but fixed number of parties, where some are characterised and some are not. Characterised parties then describe their local systems by means of quantum mechanics, i.e. the marginal states of their systems is specified by a density operator to which they have access. On the other hand, uncharacterised parties only rely on the classical labels of the inputs and outputs of their devices, and their outcome statistics. As such, in a scenario with $n$ uncharacterised parties, the object of interest is the multipartite \textit{assemblage} $\bsig_{A_1\ldots A_{n} | X_1 \ldots X_{n}} := \{ \sigma_{a_1 \ldots a_{n} | x_1 \ldots x_{n}} \}_{a_1, \ldots ,a_{n} , x_1, \ldots, x_{n}}$, the ensemble of unnormalised states $\sigma_{a_1 \ldots a_{n} | x_1 \ldots x_{n}}$, which are conditionally prepared for the characterised parties by the uncharacterised ones, 
when they input $x_1 \ldots x_{n}$ on their devices and obtain outcomes $a_1 \ldots a_{n}$ (see Fig.~\ref{fig:setup}). 
Analogously to the bipartite setting, $\tr{\sigma_{a_1 \ldots a_{n} | x_1 \ldots x_{n}}}=p(a_1 \ldots a_{n} | x_1 \ldots x_{n})$. In the following we will consider the case where there is only one characterised party, referred to as Bob. In general, our results will also apply to the case of more than one characterised party, by considering these as just one (larger) effective characterised party. We will explicitly discuss the details when the number of characterised parties plays a relevant role.

\begin{figure}
\begin{center}
\begin{tikzpicture}[scale=0.7]

\shade[draw, thick, ,rounded corners, inner color=white,outer color=gray!50!white] (-4.25,0.7) rectangle (-2.75,2) ;
\node at (-1.75,1.35) {$\cdots$};
\shade[draw, thick, ,rounded corners, inner color=white,outer color=gray!50!white] (-0.75,0.7) rectangle (0.75,2) ;
\node at (1.75,1.35) {$\cdots$};
\shade[draw, thick, ,rounded corners, inner color=white,outer color=gray!50!white] (2.75,0.7) rectangle (4.25,2) ;

\draw[thick, ->] (0,2.7) -- (0,2);
\draw[thick, ->] (-3.5,2.7) -- (-3.5,2);
\draw[thick, ->] (3.5,2.7) -- (3.5,2);

\draw[thick, ->] (0,0.7) -- (0,0);
\draw[thick, ->] (-3.5,0.7) -- (-3.5,0);
\draw[thick, ->] (3.5,0.7) -- (3.5,0);

\node at (0,3) {$x_k$};
\node at (0,-0.3) {$a_k$};
\node at (-3.5,3) {$x_1$};
\node at (-3.5,-0.3) {$a_1$};
\node at (3.5,3) {$x_n$};
\node at (3.5,-0.3) {$a_n$};

\node at (7, 2) {$\sigma_{a_1\ldots a_n|x_1\ldots x_n}$};

\node[draw=gray!60!black,shape=circle,shading=ball, ball color=gray,scale=.5] at (7,1.2) {};

\end{tikzpicture}
\end{center}
\caption{Steering scenario with $n+1$ distant parties: $n$ steering parties each having access to an uncharacterised measuring device (box) and one steered party having a characterised quantum system with full quantum control. Each steering party performs a measurement $x_k$ on their device, obtaining an outcome $a_k$.
The characterised party's systems are steered into the conditional states $\sigma_{a_1\ldots a_n|x_1\ldots x_n}$ with probability $p(a_1\ldots a_n|x_1\ldots x_n)=\tr{\sigma_{a_1\ldots a_n|x_1\ldots x_n}}$. }
\label{fig:setup}
\end{figure}
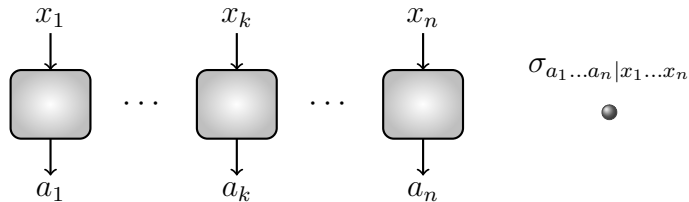

Multipartite steering experiments lead to richer phenomena than the bipartite experiments. In the former case it is possible to have steering that goes beyond what quantum mechanics allows for, while still complying with the principle of no superluminal signalling \cite{PQP}, while in the latter case this is impossible \cite{ghjw}. One of the primary goals of this paper is to develop a formalism which can deal with both quantum and post-quantum steering in a unified manner. 
To that end, in the next section we introduce a representation of general multipartite assemblages in terms of affine combinations of local hidden states.
This is a generalisation of similar results in Bell scenarios \cite{tony}, and will be useful for us to introduce a general formalism for steering later. 

\section{Pseudo LHS models}\label{ap:negLHS}

In this section, we present a characterisation of general (i.e. no-signalling) assemblages as affine combinations of local hidden states. We will denote these by \textit{pseudo-LHS models}. 

Consider hence a general steering scenario where $n$ uncharacterised parties, henceforth denoted as Alices, steer a characterised one, denoted as Bob. Assume that each of the $n$ Alices operates a device whose input can assume $m$ different values and returns one out of $d$ outcomes, hence producing the assemblage $\bsig_{A_1\ldots A_{n} | X_1 \ldots X_{n}}$ for Bob.

Whenever the Alices and Bob share a classically correlated system, the assemblages that may arise by the Alices performing  local measurements on their share of the system are said to have an LHS model, as mentioned in the previous section. The formal definition of such a model in the multipartite scenario is the following. 

\begin{defn}\label{lhs}\textbf{ LHS model.} \\
An assemblage $\bsig_{A_1\ldots A_{n} | X_1 \ldots X_{n}}$  has an LHS model if it can be decomposed as
\begin{equation}\label{lhsqm2}
\sigma_{a_1 \ldots a_{n} | x_1 \ldots x_{n}} = \sum_{\lambda} \, p_{\lambda}^{(1)}(a_1|x_1) \, \ldots p_{\lambda}^{(n)}(a_{n}|x_{n}) \, \sigma_{\lambda}
\end{equation}
where $p_{\lambda}^{(j)}(a_j|x_j) \geq 0$  is a conditional probability distribution for every $\lambda$ and every uncharacterised party $j$, and $\sigma_\lambda$ (the local hidden states) are unnormalised quantum states that satisfy
\begin{align}
&\sigma_\lambda \geq 0 \quad \forall \, \lambda \,,\\
&\tr{\sum_{\lambda} \, \sigma_{\lambda}} =1.
\end{align}
\end{defn}

The purpose of this work is to develop a general framework for steering that goes beyond LHS and quantum assemblages. A possible strategy for this is to generalise the definition of an LHS model to include quantum assemblages and potentially some post-quantum ones, in a similar spirit as previously done in non-locality \cite{tony, degorre}. Thus, we propose the following generalisation, which we denote as \textit{pseudo-LHS models}.

\begin{defn}\label{qlhs}\textbf{Pseudo-LHS model.} \\
An assemblage $\bsig_{A_1\ldots A_{n} | X_1 \ldots X_{n}}$  has a pseudo-LHS model if it can be decomposed as
\begin{equation}\label{qm2}
\sigma_{a_1 \ldots a_{n} | x_1 \ldots x_{n}} = \sum_{\lambda} \, p_{\lambda}^{(1)}(a_1|x_1) \, \ldots p_{\lambda}^{(n)}(a_{n}|x_{n}) \, \sigma_{\lambda}
\end{equation}
where $p_{\lambda}^{(j)}(a_j|x_j) \geq 0$  is a conditional probability distribution for every $\lambda$ and every uncharacterised party $j$, and the LHSs satisfy
\begin{equation}
\tr{\sum_{\lambda} \, \sigma_{\lambda}} =1.
\end{equation}
\end{defn}

Note that in Def. \ref{qlhs}, if we demand in addition that $\sigma_\lambda \geq 0$ $\forall \, \lambda$, we recover Def. \ref{lhs} of a LHS model. Hence, we are relaxing the model by allowing local hidden states that are not positive semidefinite. In particular, this implies that we allow the hidden variables $\lambda$ to have negative probabilities, since $p(\lambda)=\tr{\sigma_\lambda}$\footnote{A natural question is what would happen if the local hidden states are allowed to not be positive semidefinite but constrained to $p(\lambda) \geq 0$. The set of assemblages that admit such a model is strictly contained within the pseudo-LHS set, since they only allow for local correlations for the output statistics of the uncharacterised parties. }. 

Note however that, when generalising LHS models we encounter a freedom that was not present in Bell scenarios. Indeed, from Eq.~\eqref{lhsqm2} one could either relax the LHS assumption by considering assemblages that are (i) convex combinations of non-positive semidefinite states, or (ii) affine combinations of positive semidefinite states. Definition \ref{qlhs} corresponds to (i). In Bell scenarios, in contrast, the corresponding formalism admits only the analogue to (ii), in terms of affine combinations of local correlations. This freedom, however, does not introduce any ambiguity in the formalism since they turn out to be equivalent, as we show next. 

\begin{lem}\label{thm-NS-affine-prod}
Let $\bsig_{A_1\ldots A_{n} | X_1 \ldots X_{n}}$ be an assemblage in a steering scenario where $n$ uncharacterised parties steer a characterised one. The assemblage has a pseudo-LHS model iff it can be written as an affine combination of quantum states. 
\end{lem}
\begin{proof}
 First, consider an assemblage that has a decomposition as an affine combination of quantum states:
 \begin{equation}
 \sigma_{a_1 \ldots a_{n} | x_1 \ldots x_{n}} = \sum_{\lambda} \, q(\lambda) \, p_{\lambda}^{(1)}(a_1|x_1) \, \ldots p_{\lambda}^{(n)}(a_{n}|x_{n}) \, \rho_\lambda\,,
 \end{equation}
 where $\rho_\lambda$ are, for each $\lambda$, normalised hidden quantum states on Bob's system and $q(\lambda)$ is a pseudo probability distribution on $\lambda$, i.e. $q(\lambda) \in \mathbb{R}$ for all $\lambda$ and $\sum_\lambda q(\lambda) = 1$. By defining $\sigma_\lambda := q(\lambda)  \rho_{\lambda}$ it follows that the assemblage has a pseudo-LHS model. 
 
 For the converse, start from an assemblage with a pseudo-LHS model:
 \begin{align}\label{eq:con1}
 \sigma_{a_1 \ldots a_{n} | x_1 \ldots x_{n}} = \sum_{\lambda} \, p_{\lambda}^{(1)}(a_1|x_1) \, \ldots p_{\lambda}^{(n)}(a_{n}|x_{n}) \, \sigma_{\lambda}\,.
 \end{align}
Each $\sigma_\lambda$ can be expressed as
\begin{align*}
\sigma_\lambda = c_{\lambda,+} \, \rho_{\lambda,+} - c_{\lambda,-} \rho_{\lambda,-} \, \quad \forall \, \lambda\,,
\end{align*}
where the operators $\rho_{\lambda,+}$ and $\rho_{\lambda,-}$ are normalised quantum states and $c_{\lambda,+}$ and $c_{\lambda,-}$ non-negative reals such that $p(\lambda)=c_{\lambda,+} - c_{\lambda,-}$ for all $\lambda$. 

By introducing an auxiliary binary hidden variable $\mu = \{+,-\}$, Eq.~\eqref{eq:con1} may be rewritten as
\begin{align}
 \sigma_{a_1 \ldots a_{n} | x_1 \ldots x_{n}} &= \sum_{\lambda, \mu} \, p_{\lambda}^{(1)}(a_1|x_1) \, \ldots p_{\lambda}^{(n)}(a_{n}|x_{n}) \, q(\lambda, \mu)\,\rho_{\lambda,\mu}\,, \label{NS-affine-prod}
\end{align}
where $q(\lambda, \mu) := \mu c_{\lambda,\mu}$. The fact that this is a pseudo probability distribution on $\lambda$ and $ \mu$ follows from the fact that $\sum_{\lambda, \mu} q(\lambda, \mu) = \sum_{\lambda} \, \tr{\sigma_{\lambda}} =1$. Hence, the assemblage may be written as an affine combination of normalised quantum states. 
\end{proof}

This allows us to understand the problem in a semi-classical way (see Fig.~\ref{f:negmod}). An unphysical source produces the hidden variables $(\lambda, \mu)$ with pseudo probability $q(\lambda, \mu)$ \cite{feynman} and sends them to the $n+1$ parties. The uncharacterised parties produce the outcomes via the response functions $p_{\lambda}^{(j)}(a_j|x_j)$ whereas the characterised one produces locally the states $\rho_{\lambda, \mu}$. The assemblage is then explained by Eq.~(\ref{NS-affine-prod}) as just an affine combination of such semi-classical preparations. Note that all the non-classicality of the assemblage is contained in the negativity of the pseudo-probability distribution $q$. 
In the case where the steering scenario consists of more than one characterised party (say, $t$), we can take a step further and express each of the quantum states $\rho_{\lambda, \mu}$ as affine combinations of product states $\rho_{\lambda, \mu, \nu} :=\rho_{\lambda, \mu, \nu}^{(1)} \otimes \ldots \otimes\rho_{\lambda, \mu, \nu}^{(t)}$ with pseudo-probabilities $q(\lambda, \mu, \nu)$ \cite{jonGPT}. Hence, the assemblage may in this case be expressed as
\begin{align}
 \sigma_{a_1 \ldots a_{n} | x_1 \ldots x_{n}} &= \sum_{\lambda, \mu, \nu} \, p_{\lambda}^{(1)}(a_1|x_1) \, \ldots p_{\lambda}^{(n)}(a_{n}|x_{n}) \, q(\lambda, \mu, \nu)\,\rho_{\lambda,\mu, \nu}\,. \label{NS-affine-prod-many}
 \end{align}
This generalises the possibility to express as affine combinations both conditional probability distributions for non-locality scenarios \cite{tony} and shared quantum states \cite{jonGPT}. 
A similar semi-classical interpretation of the steering experiment with many characterised parties is presented in Fig.~\ref{f:negmod}b.

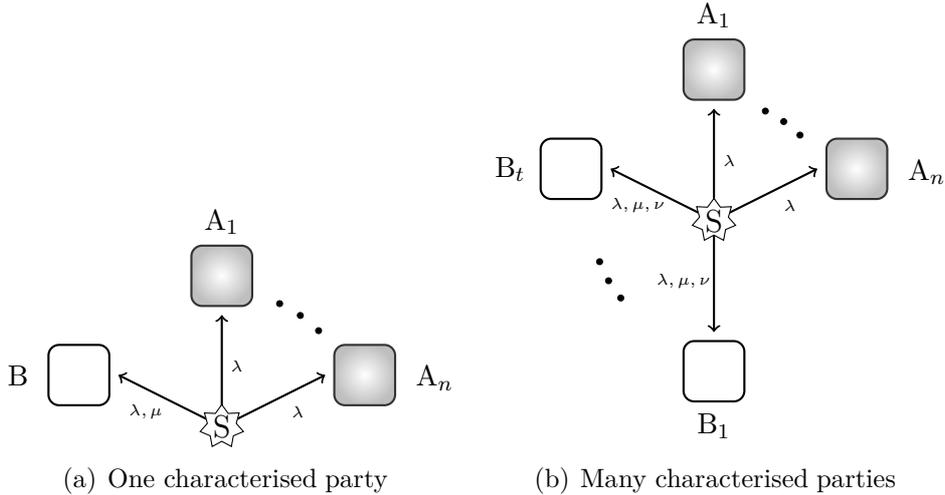
\begin{figure}
\begin{center}
\subfigure[One characterised party]{
\begin{tikzpicture}[scale=0.4]
\node (ai) at (90:5) {} ;
\node (af) at (20:5) {} ;
\node (bf) at (160:5) {} ;

\node at ($ (ai) + (0,1.8) $) {\small{$\mathrm{A}_1$}};
\node at ($ (af) + (2.3,-0.1) $) {\small{$\mathrm{A}_{n}$}};
\node at ($ (bf) + (-2,0) $) {\small{$\mathrm{B}$}};

\shade[draw, thick, color=gray!40!black ,rounded corners, inner color=white,outer color=gray!50!white] ($ (ai) + (-1,-1) $) rectangle ($ (ai) + (1,1) $) ;
\shade[draw, thick,color=gray!40!black ,rounded corners, inner color=white,outer color=gray!50!white] ($ (af) + (-1,-1) $) rectangle ($ (af) + (1,1) $) ;
\draw[thick, ,rounded corners] ($ (bf) + (-1,-1) $) rectangle ($ (bf) + (1,1) $) ;
\node[draw,shape=circle,fill,scale=.2] at (45:4.5) {};
\node[draw,shape=circle,fill,scale=.2] at (55:4.5) {};
\node[draw,shape=circle,fill,scale=.2] at (65:4.5) {};

\draw[thick, ->] (0,0) -- ($ (ai) + (0,-1.3) $);
\draw[thick, ->] (0,0) -- ($ (af) + (-1.3,0) $);
\draw[thick, ->] (0,0) -- ($ (bf) + (1.3,0) $);

\node[draw, shape=star,star points=7, fill,color=black,inner color=white,scale=1 ] at (0,0) {};
\node at (0,0) {S};

\node at (2.5,0.5) {\tiny{$\lambda$}};
\node at (0.5,2) {\tiny{$\lambda$}};
\node at (-2.5,0.5) {\tiny{$\lambda,\mu$}};
\end{tikzpicture}}
\subfigure[Many characterised parties]{
\begin{tikzpicture}[scale=0.4]
\node (ai) at (90:5) {} ;
\node (af) at (20:5) {} ;
\node (bi) at (270:5) {} ;
\node (bf) at (160:5) {} ;

\node at ($ (ai) + (0,1.8) $) {\small{$\mathrm{A}_1$}};
\node at ($ (af) + (2.3,-0.1) $) {\small{$\mathrm{A}_{n}$}};
\node at ($ (bi) + (0,-1.8) $) {\small{$\mathrm{B}_1$}};
\node at ($ (bf) + (-2,0) $) {\small{$\mathrm{B}_t$}};

\shade[draw, thick, color=gray!40!black ,rounded corners, inner color=white,outer color=gray!50!white] ($ (ai) + (-1,-1) $) rectangle ($ (ai) + (1,1) $) ;
\shade[draw, thick,color=gray!40!black ,rounded corners, inner color=white,outer color=gray!50!white] ($ (af) + (-1,-1) $) rectangle ($ (af) + (1,1) $) ;
\draw[thick, ,rounded corners] ($ (bi) + (-1,-1) $) rectangle ($ (bi) + (1,1) $) ;
\draw[thick, ,rounded corners] ($ (bf) + (-1,-1) $) rectangle ($ (bf) + (1,1) $) ;
\node[draw,shape=circle,fill,scale=.2] at (45:4) {};
\node[draw,shape=circle,fill,scale=.2] at (55:4) {};
\node[draw,shape=circle,fill,scale=.2] at (65:4) {};
\node[draw,shape=circle,fill,scale=.2] at (210:4) {};
\node[draw,shape=circle,fill,scale=.2] at (200:4) {};
\node[draw,shape=circle,fill,scale=.2] at (220:4) {};

\draw[thick, ->] (0,0) -- ($ (ai) + (0,-1.3) $);
\draw[thick, ->] (0,0) -- ($ (af) + (-1.3,0) $);
\draw[thick, ->] (0,0) -- ($ (bi) + (0,1.3) $);
\draw[thick, ->] (0,0) -- ($ (bf) + (1.3,0) $);

\node[draw, shape=star,star points=7, fill,color=black,inner color=white,scale=1 ] at (0,0) {};
\node at (0,0) {S};

\node at (2.5,0.5) {\tiny{$\lambda$}};
\node at (0.5,2) {\tiny{$\lambda$}};
\node at (-1,-2) {\tiny{$\lambda,\mu, \nu$}};
\node at (-2.5,0.5) {\tiny{$\lambda,\mu, \nu$}};
\end{tikzpicture}}
\end{center}
\caption{\textbf{Semi-classical approach to a no-signalling assemblage.} 
(a) One characterised party: An unphysical source produces the hidden variables $(\lambda, \mu)$ with pseudo probability $q(\lambda, \mu)$ and sends them to the $n+1$ parties. The uncharacterised parties produce the outcomes via the response functions $p_{\lambda}(a_j|x_j)$, whereas the characterised ones produce the states $\rho_{\lambda, \mu}$ locally. The no-signalling assemblage is then explained by Eq.~(\ref{NS-affine-prod}) as an affine combination of such local preparations.
(b) $t$ characterised parties: A source produces the hidden variables $(\lambda, \mu, \nu)$ with pseudo probability $q(\lambda, \mu, \nu)$ and sends them to the $n+t$ parties. The uncharacterised parties produce the outcomes via the response functions $p_{\lambda}(a_j|x_j)$ whereas the characterised ones produce locally the states $\rho_{\lambda, \mu, \nu}^{(i)}$. The non-signalling assemblage is then explained by Eq.~(\ref{NS-affine-prod-many}) as an affine combination of such local preparations. 
In both (a) and (b), all the non-classicality of the assemblage is contained in the negativity of the pseudo-probability distribution $q$. }
\label{f:negmod}
\end{figure}


\section{No-signalling assemblages}\label{sec3}

The formalism that we present in this work provides a unified framework for the study of no-signalling assemblages in general steering scenarios. In this section we will review the basics of no-signalling assemblages and relate them to the pseudo-LHS models from the previous section. 

A general assemblage that complies with the no signalling principle is defined as follows: 

\begin{defn}\label{ns}\textbf{No-signalling assemblage.} \\
An assemblage $\bsig_{A_1\ldots A_{n} | X_1 \ldots X_{n}}$ is no-signalling if it satisfies 
\begin{align}
\sum_{a_1 \ldots a_{n}} \sigma_{a_1 \ldots a_{n}|x_1 \ldots x_{n}} &= \rho_\mathrm{R} \quad \forall \, x_1 \ldots x_{n}  \label{eq:normalisation},
\end{align}
where $\rho_\mathrm{R}$ is the (normalised) reduced state of the characterised party's system, 
and for every subset $\mathcal{S}=\{ i_1 \ldots i_r\}$ of $r$ uncharacterised parties, with $1 \leq r < n$,
\begin{align}
\sum_{{a_j ,\, j \not\in \mathcal{S}}} \sigma_{a_1 \ldots a_{n}|x_1 \ldots x_{n}} &= \sigma_{a_{i_1} \ldots a_{i_r}|x_{i_1} \ldots x_{i_r}} \quad \forall \, x_{i_1} \ldots x_{i_r}\,. \label{eq:NSall2one}
\end{align}
\end{defn}

Condition \eqref{eq:NSall2one} says that when disregarding the outcomes obtained by some uncharacterised parties, the state of the characterised party's subsystem should not depend on the choice of measurement of the disregarded parties. Moreover, when all the uncharacterised parties are traced out, condition \eqref{eq:normalisation} says that the  state of the characterised one should be a normalised quantum state equal to his subsystem's reduced state. We are now in a position to present one of our main results.

\begin{thm}\label{affine_NS}
Let $\bsig_{A_1\ldots A_{n} | X_1 \ldots X_{n}}$ be an assemblage in a steering scenario where $n$ uncharacterised parties steer a characterised one. The assemblage is no-signalling iff it has a pseudo-LHS model. 
\end{thm}
\begin{proof}
Given an assemblage with a psuedo local hidden state model, Eq. (\ref{qm2}) guarantees that it satisfies the no-signalling constraints, hence the first implication follows.

For the converse, let us assume that $\bsig_{A_1\ldots A_{n} | X_1 \ldots X_{n}}$ is no-signalling. For party $j \in \{1, \ldots, n\}$, define a local hidden variable $\lambda_j$, taking values in the set 
\begin{align}\label{eq:thelambda}
\Lambda_j = \left\{ [a_j, x_j] \right\}_{a_j,x_j} \cup \{\xi\} \,,
\end{align}
i.e. the set of ordered pairs $[a_j, x_j]$ in union with a single-element set composed of an arbitrary dummy symbol, denoted by $\xi$. There are $m\, d$ pairs $[a_j, x_j]$, so $|\Lambda_j|=m\, d+1$. 

Then, take the local hidden variable $\lambda$ of Eq.~\eqref{qm2} as the tuple $\blam := (\lambda_1, \ldots, \lambda_n)$, and in turn define the weights in decomposition  (\ref{qm2})  as
\begin{equation}\label{weights}
p_{\lambda_j}^{(j)}(a_j|x_j) = \begin{cases} \delta_{\lambda_j, [a_j, x_j]} &\mathrm{if} \quad a_j<d \\ 1-\sum_{\tilde{a}<d} \delta_{\lambda_j, [\tilde{a}, x_j]} &\mathrm{if} \quad  a_j=d.\end{cases}
\end{equation}
These are well defined conditional probability distributions of every $\lambda_j$ and party $j$, since $\sum_{a_j} p_{\lambda_j}^{(j)}(a_j|x_j) = 1$. 

Given the global hidden variable $\blam$, define $S_{\blam}$ to be the set of indices $\{j \, : \, \lambda_j \neq \xi \}$. With this, define the hidden pseudo-states as 
\begin{align}\label{mpqs}
\sigma_{\blam} &:= (1-m)^{n-|S_{\blam}|} \, \sigma_{\mathbf{a}_{S_{\blam}} | \mathbf{x}_{S_{\blam}}},
\end{align} 
where the $\mathbf{a}_{S_{\blam}}$ and $\mathbf{x}_{S_{\blam}}$ involve the parties that belong to the set $S_{\blam}$, i.e. those whose hidden variable does not take the dummy value $\xi$. 
For instance, when $|S_{\blam}|=n$, 
$$
\sigma_{[a_1,x_1],\ldots,[a_n,x_n]} = \sigma_{a_1\ldots a_n | x_1 \ldots x_n}\,,
$$ 
and when $|S_{\blam}|=n-1$ with $\lambda_1 = \xi$
$$
\sigma_{\xi,[a_2,x_2],\ldots,[a_n,x_n]} = (1-m) \,  \sigma_{a_2\ldots a_n | x_2 \ldots x_n}\,.
$$ 
Note that $\sigma_{\ba_{S_{\blam}} | \bx_{S_{\blam}}}$ is well-defined since the original assemblage is no-signalling, and $\sigma_{\ba_{S_{\blam}} | \bx_{S_{\blam}}}$ arises from it by tracing out the parties that are not in $S$.

Now we need to prove that these $\sigma_{\blam}$ are suitably normalised and that, together with the $p_{\lambda_j}^{(j)}(a_j|x_j) $ from Eq.~\eqref{weights}, they reproduce the assemblage. 
For the former:
\begin{align*}
\tr{\sum_{\blam} \sigma_{\blam} } &= \tr{\sum_{S \subseteq \{1, \dots, {n}\}} \sum_{\ba_S, \bx_S} (1-m)^{n-|S|} \, \sigma_{\ba_S | \bx_S} }\\
&= \tr{\sum_{S \subseteq \{1, \dots, {n} \}} (1-m)^{n-|S|} m^{|S|}\rho_{\mathrm{R}}} \\
&=\sum_{S \subseteq \{1, \dots, n\}} (1-m)^{n-|S|} m^{|S|}\\
&= \sum_{r = 0}^{n} \binom{n}{r} (1-m)^{n-r} m^{r}\\
&= 1.
\end{align*}

For the last part, we will first show that this pseudo-LHS model recovers the assemblage for the cases when $a_j <d$ for every party $j$. Then we will prove that the statement also holds when some of the outcomes have value $d$, by induction.
For the former, by definition,
\begin{align*}
\sum_{\blam} \, p_{\lambda_1}^{(1)}(a_1|x_1) \, \ldots p_{\lambda_{n}}^{(n)}(a_{n}&|x_{n}) \, \sigma_{\blam} = \sum_{\blam} \,\delta_{\lambda_1, [a_1, x_1]}  \, \ldots \, \delta_{\lambda_{n}, [a_{n}, x_{n}]} \, \sigma_{\blam}\,.
\end{align*}
Since $\lambda_j = \xi \Rightarrow \delta_{\lambda_j, [a_j, x_j]} = 0$, the only non-trivial terms in the sum are those where no party's hidden variable takes the dummy value, and for those, $\sigma_{[a_1,x_1] \ldots [a_{n},x_{n}]} = \sigma_{a_1 \ldots a_n | x_1 \ldots x_n}$. Therefore, 
\begin{align*}
\sum_{\blam} \,\delta_{\lambda_1, [a_1, x_1]}  \, \ldots  \, \delta_{\lambda_{n}, [a_{n}, x_{n}]} \, \sigma_{\blam} &= 
\sum_{[\tilde{a}_1,\tilde{x}_1],\ldots,[\tilde{a}_{n},\tilde{x}_{n}]} \,\delta_{[\tilde{a}_1,\tilde{x}_1], [a_1, x_1]}  \, \ldots \, \delta_{[\tilde{a}_{n},\tilde{x}_{n}], [a_{n}, x_{n}]} \, \sigma_{\tilde{a}_1 \ldots \tilde{a}_{n} | \tilde{x}_1 \ldots \tilde{x}_{n}} \\
&=\sigma_{a_1 \ldots a_{n} | x_1 \ldots x_{n}}.
\end{align*}
Now let us assume that the pseudo-LHS model reproduces the assemblage when the first $k$ uncharacterised parties obtain outcomes $a_j=d$. That is, we assume that
\begin{align}\label{eq:k}
\sigma_{\!\!\underbrace{{\scriptstyle d \ldots d}}_{k \text{ times}} \!\!\! a_{k+1} \ldots a_n | x_1 \ldots x_n} = \sum_{\blam}\, p_{\lambda_1}^{(1)}(d|x_1) \,  \ldots p_{\lambda_k}^{(k)}(d|x_k) p_{\lambda_{k+1}}^{(k+1)}(a_{k+1}|x_{k+1}) \, \ldots  p_{\lambda_{n}}^{(n)}(a_{n}&|x_{n}) \, \sigma_{\blam} .
\end{align}
Hence, when now the first $k+1$ parties obtain outcome $d$ it follows that:
\begin{align}\label{eq:k+1}
&\sum_{\blam}\, \prod_{j=1}^{k+1}p_{\lambda_j}^{(j)}(d|x_j)\prod_{l=k+2}^{n}p_{\lambda_{l}}^{(l)}(a_{l}|x_{l})\sigma_{\blam} \nonumber\\
=&\sum_{\blam}\, \prod_{j=1}^{k}p_{\lambda_j}^{(j)}(d|x_j)\left( 1 - \sum_{\tilde{a}_{k+1} < d} p_{\lambda_{k+1}}^{(k+1)}(\tilde{a}_{k+1}|x_{k+1})  \right) \prod_{l=k+2}^{n}p_{\lambda_{l}}^{(l)}(a_{l}|x_{l}) \sigma_{\blam} \nonumber\\
=&\sum_{\blam}\, \prod_{j=1}^{k}p_{\lambda_j}^{(j)}(d|x_j)\prod_{l=k+2}^{n}p_{\lambda_{l}}^{(l)}
(a_{l}|x_{l}) \sigma_{\blam} - \sum_{\tilde{a}_{k+1} < d} \sigma_{\!\!\underbrace{{\scriptstyle d \ldots d}}_{k \text{ times}} \!\!\! \tilde{a}_{k+1} \ldots a_n | x_1 \ldots x_n} \nonumber\\
=&m\, \sigma_{\!\!\underbrace{{\scriptstyle d \ldots d}}_{k \text{ times}} \!\!\! a_{k+2} \ldots a_n | x_1 \ldots x_k x_{k+2} \ldots  x_n} + (1-m)\,\sigma_{\!\!\underbrace{{\scriptstyle d \ldots d}}_{k \text{ times}} \!\!\! a_{k+2} \ldots a_n | x_1 \ldots x_k x_{k+2} \ldots  x_n} - \sum_{\tilde{a} < d} \sigma_{\!\!\underbrace{{\scriptstyle d \ldots d}}_{k \text{ times}} \!\!\! \tilde{a}_{k+1} \ldots a_n | x_1 \ldots x_n} \nonumber\\
=&\sigma_{\!\!\!\!\underbrace{{\scriptstyle d \ldots d}}_{k+1 \text{ times}} \!\!\!\!\!\! a_{k+1} \ldots a_n | x_1 \ldots x_n}\,,
\end{align}
This may be understood as follows: if the pseudo-LHS model reproduces the elements of the assemblage where the first $k$ uncharacterised parties obtain outcome $d$ (Eq.~\eqref{eq:k}), it also reproduces the elements of the assemblage where the first $k+1$ uncharacterised parties obtain outcome $d$ (Eq.~\eqref{eq:k+1}), and this holds for any value of $k$. 
The argument for any other subset of $k+1$ parties that is not necessarily $\{ 1, \ldots , k+1\}$ follows similarly. 

Therefore, if the pseudo-LHS model reproduces the elements of the assemblage where $k$ uncharacterised parties obtain outcome $d$, it also reproduces the elements of the assemblage where $k+1$ uncharacterised parties have outcome $d$. By induction it follows that the pseudo-LHS model recovers the full assemblage for every value of $a_1 \ldots a_n$. 
\end{proof}

\section{A formalism for non-signalling steering}
\label{se:for}

In this section, we develop a formalism for non-signalling steering, similar to the one presented in \cite{toni} (see also \cite{barnum}) for non-signalling correlations in Bell scenarios. 

\begin{thm}\label{theOteo}
Let $\bsig_{A_1\ldots A_{n} | X_1 \ldots X_{n}}$ be an assemblage in a steering scenario where $n$ uncharacterised parties steer a characterised one (labelled $B$). The assemblage is no-signalling iff there exist  POVM elements $M^{(j)}_{a_j|x_j}$ for each uncharacterised party $j$ (i.e. positive operators satisfying $\sum_{a_j} M^{(j)}_{a_j|x_j} = \mathbbm{1}$) and a unit trace Hermitian operator $\tilde{\rho}$ such that:
\begin{equation}\label{Ogen}
\sigma_{a_1 \ldots a_{n} | x_1 \ldots x_{n}} = \Tr{ \left( M^{(1)}_{a_1|x_1} \otimes \ldots \otimes M^{(n)}_{a_{n}|x_{n}} \otimes \mathbbm{1} \right) \, \tilde{\rho}}{1, \ldots n},
\end{equation}
where the partial trace involves the $n$ uncharacterised subsystems (see Fig.~\ref{fig:processing}.(a)).
\end{thm}
\begin{proof}
If an assemblage can be written as in Eq. (\ref{Ogen}), it is straightforward to see that it is no-signalling. The `only-if' part of the proof relies on the constructions of Thm. \ref{affine_NS} and Lem. \ref{thm-NS-affine-prod}, as we explicitly show in what follows. 

First, write the no-signalling assemblage as an affine combination of quantum states, as in Lemma \ref{thm-NS-affine-prod} by further using the hidden variable model from Theorem \ref{affine_NS}: 
\begin{align}
\sigma_{a_1 \ldots a_{n} | x_1 \ldots x_{n}} = \sum_{\blam, \mu} q(\blam, \mu)p_{\lambda_1}^{(1)}(a_1|x_1) \ldots p_{\lambda_n}^{(n)}(a_n|x_n)  \,\rho_{\blam, \mu} \,,
\end{align}
with $p_{\lambda_j}^{(j)}(a_j|x_j) $ as in Eq. (\ref{weights}). 

Then, assign to each uncharacterised party $j$ an $(md+1)$-dimensional Hilbert space spanned by the orthonormal basis $\{ \ket{\lambda_j} \, : \, \lambda_j \in \Lambda_j\}$, where $\Lambda_j$ defined in Eq.~\eqref{eq:thelambda} is the set of values that the hidden variable for party $j$ can take. Define
\begin{equation}\label{trho}
\tilde{\rho} := \sum_{\blam,\mu}\, q(\blam, \mu) \, \ket{\lambda_1} \bra{\lambda_1} \otimes \ldots \otimes \ket{\lambda_{n}} \bra{\lambda_{n}} \otimes\rho_{\blam, \mu} \,,
\end{equation}
and
\begin{equation}
M^{(j)}_{a_j|x_j}  := \sum_{\lambda_j}   p_{\lambda_j}^{(j)}(a_j|x_j)  \ket{\lambda_j} \bra{\lambda_j} \,.
\end{equation}

Since the $\{ \ket{\lambda_j} \}$ bases are orthonormal, it follows by direct calculation that one correctly obtains a pseudo-LHS model for the desired assemblage. 
\end{proof}

Here, this Hermitian operator $\tilde{\rho}$ plays the role of the operator $O$  in \cite{toni}. 
Note that for a given assemblage, the choice of $\tilde{\rho}$ is not unique. The construction presented in Theorem \ref{theOteo} produces a specific $\tilde{\rho}$ which works in all situations. 

By definition, $\tilde{\rho}$ in Eq.~(\ref{Ogen}) can be chosen to be positive semidefinite if and only if the assemblage is quantum. 
On the other hand, it also follows that an assemblage has an LHS model if and only if  $\tilde{\rho}$  can be chosen to be a fully $(n+1)$-separable quantum state across the multipartition $A_1 | \ldots | A_n | B$.

Once the nature of the assemblages is identified with the properties of the operator $\tilde{\rho}$, one can study the families of assemblages for different families of $\tilde{\rho}$ that have particular properties. Of particular interest is the set of \textit{Gleason assemblages}, which contains the set of quantum assemblages. 
In analogy with the Gleason correlations of Ref. \cite{toni}, Gleason assemblages are those that arise when $\tilde{\rho}$ can be chosen to be a Hermitian operator $W$ that produces well-defined local measurements by the uncharacterised parties. This is a stronger requirement than that imposed in general by Theorem~\ref{theOteo}, where the operator $\tilde{\rho}$ need only produce valid assemblages for the specific measurements $M_{a_j|x_j}^{(j)}$. 
A necessary and sufficient condition is that $W$ is an entanglement witness with respect to the $(n+1)$-partition $\mathrm{A}_1\, | \, \mathrm{A}_2 \, | \, \ldots \, | \, \mathrm{A}_{n} \, | \,B$ \cite{GTB05,AMD15}. 

Since we demand that $W$ defines valid assemblages for \textit{all} local measurements (not just some particular subset of measurement), the set of Gleason assemblages is in general smaller than the no-signalling set. Also, since $W$ may be non-positive, the set of Gleason assemblages is in general larger than the quantum set.

One can see that for bipartite steering scenarios, the set of Gleason assemblages coincides with both the quantum and the no-signalling set. Following Ref.~\cite{toni}, this can be seen by considering that any unit-trace bipartite entanglement witness $W_{AB}$ can be expressed as the action on the steering side of a trace-preserving positive map $\mathcal{E}$ on a bipartite normalized quantum state, $W_{AB}=(\mathcal{E}_A\otimes\id_B)[\rho_{AB}]$. Hence,
\begin{align*}
\sigma_{a|x}
&=\Tr{(M_{a|x}\otimes \openone_B )W_{AB}}{A}\\
&=\Tr{(\mathcal{E}^\dagger [M_{a|x}]\otimes \openone_B )\rho_{AB}}{A}\\
\end{align*}
with $\mathcal{E}^\dagger$, the dual of $\mathcal{E}$, a positive unital map, so that $\{(\mathcal{E}^\dagger [M_{a|x}]\}_{a}$ is also a POVM for all $x$. 
However, for steering scenarios with more than one uncharacterised party this is no longer the case, as we see next.

\begin{ex}\label{ex:wit}
Consider the four three-qubit states:
\begin{equation}
\label{eq:upb_four_states}
\ket{000}, \quad \ket{1e^{\perp} e}, \quad \ket{e1e^\perp}, \quad \ket{e^\perp e 1}\,,
\end{equation}
where $\{\ket{e}, \ket{e^\perp}\}$ is an arbitrary basis different from $\{\ket{0},\ket{1}\}$. Denote by $\Pi_{\mathrm{UPB}}$ the projector onto the subspace spanned by all four states in Eq.~\eqref{eq:upb_four_states}. Construct now the tripartite entanglement witness 
\begin{equation}
W = \frac{1}{4-8\epsilon} \, \left( \Pi_{\mathrm{UPB}} - \epsilon \mathbbm{1} \right),
\end{equation}
where $\epsilon = \min_{\ket{\alpha \beta \gamma}} \bra{\alpha \beta \gamma}\Pi_{\mathrm{UPB}} \ket{\alpha \beta \gamma}$, with $\ket{\alpha}$, $\ket{\beta}$, and $\ket{\gamma}$ arbitrary single qubit states and $\ket{\alpha \beta \gamma}:=\ket{\alpha}\otimes\ket{\beta}\otimes\ket{\gamma}$. 
Define now the assemblage:
\begin{equation}
\sigma_{a_1a_2|x_1x_2} := \mathrm{tr}_{1,2} \left\{ (M^{(1)}_{a_1|x_1} \otimes  M^{(2)}_{a_2|x_2}  \otimes \mathbbm{1}) W     \right\},
\end{equation}
where $M^{(i)}_{1|1}=\ket{0}\bra{0}$, $M^{(i)}_{2|1}=\ket{1}\bra{1}$, $M^{(i)}_{1|2}=\ket{e}\bra{e}$, and $M^{(i)}_{2|2}=\ket{e^\perp}\bra{e^\perp}$, for $i=1,2$.

This assemblage is post-quantum, since by Bob performing measurements in the same basis as the Alices, one obtains supra-quantum correlations $p(a_1,a_2,b|x_1,x_2,y)$, as proven in \cite{toni}. Hence, already for the simplest multipartite case, the set of Gleason assemblages is larger than the quantum one. 
\end{ex}

\begin{rem}
Consider an arbitrary Gleason assemblage in a steering scenario where two uncharacterised parties steer a characterised one. This has the form 
$$
\sigma_{a_1a_2|x_1x_2} = \mathrm{tr}_{1, 2} \left\{( M^{(1)}_{a_1|x_1} \otimes  M^{(2)}_{a_2|x_2} \otimes  \mathbbm{1} )W     \right\}.
$$ 
If we now trace out the steered party we have that
$$p(a_1,a_2|x_1,x_2) = \mathrm{tr} \left\{( M^{(1)}_{a_1|x_1} \otimes  M^{(2)}_{a_2|x_2} ) W^{(12)} \right\},$$ where $W^{(12)}$ is an entanglement witness for Alice's two subsystems. Such $p(a_1,a_2|x_1,x_2)$ belong to the so called set of Gleason correlations \cite{toni}, which for bipartite Bell scenarios coincides with quantum correlations. Hence, $p(a_1,a_2|x_1,x_2)$ are quantum correlations.

Therefore, we see that Gleason assemblages, even if post-quantum, only generate quantum correlations between the two uncharacterised parties. 
Note however that when considering the full tripartite Bell scenario that includes Bob (i.e. not tracing him out) the correlations may be post-quantum. Hence, the post-quantumness of the assemblage may nevertheless be certified in a Bell experiment. 
\end{rem}

A natural question is whether post-quantum steering is a new phenomenon in its own right, or if it is just another consequence of post-quantum non-locality. In Ref. \cite{PQP}, the authors show the former to be the case. They gave an example of a post-quantum assemblage in a tripartite steering scenario with two uncharacterised parties, which cannot give rise to post-quantum non-locality in a tripartite Bell scenario, where the characterised party performs any set of measurements on their system. 

In the next section, we  use the relation between Gleason assemblages and entanglement witnesses to provide a general construction for post-quantum steering that never gives rise to post-quantum nonlocal correlations. This is the first general construction of this type, and sheds the first light on the structure underlying post-quantum steering without post-quantum nonlocality. 

\section{Post-quantum steering from positive maps}
\label{sec5}

Here we present a method for generating post-quantum assemblages without post-quantum Bell non-locality. The insight we use is the fact that positive, but not completely positive, maps are in correspondence with entanglement witnesses. 
We will see that starting from this perspective, we can identify a subset of Gleason assemblages which cannot give rise to post-quantum Bell non-locality. Furthermore, by checking simple examples of positive maps, we find that we indeed produce post-quantum steering, and hence that there is a link between positive maps and post-quantum steering. 

We may obtain a Hermitian operator $\tilde{\rho}$ to be used in Eq.~\eqref{Ogen} by acting partially on a quantum state with a positive trace-preserving (PTP) map that is not completely positive (CP). More in detail, consider a quantum state $\rho$ shared by $n+1$ parties, and define the map:
\begin{align*}
\mathcal{E}[\cdot ] := \mathcal{I}^{(1)} \otimes \cdots \otimes \mathcal{I}^{(n)} \otimes \Lambda^{(B)} [\cdot]\,,
\end{align*}
where $\Lambda^{(B)}[\cdot]$ is a PTP map. If $\Lambda^{(B)}[\cdot]$ is not CP, $\mathcal{E}[\rho]$ may be not positive semi-definite. Nevertheless, the conditional states
\begin{align}\label{eq:assemchan}
\sigma_{a_1 \ldots a_n | x_1 \ldots x_n}
&:= \Tr{( M^{(1)}_{a_1|x_1} \otimes \ldots \otimes M^{(n)}_{a_{n}|x_{n}} \otimes \id^{(B)}) \mathcal{E}[\rho]}{1, \ldots n}\\
&=\Lambda^{(B)}\left[\sigma^{\rm{Q}}_{a_1 \ldots a_n | x_1 \ldots x_n}\right]
\end{align}
form a well-defined assemblage (i.e. with $\sigma_{a_1 \ldots a_n | x_1 \ldots x_n}\geq0$). Here, $\sigma^{\rm{Q}}_{a_1 \ldots a_n | x_1 \ldots x_n} := \Tr{( M^{(1)}_{a_1|x_1} \otimes \ldots \otimes M^{(n)}_{a_{n}|x_{n}} \otimes \id^{(B)} ) \rho}{1, \ldots, n}$ are the elements of the assemblage obtained by the measurements of the Alices acting on $\rho$ rather than on $\mathcal{E}[\rho]$, and by construction they constitute a quantum assemblage. In other words, assemblages $\bsig_{A_1\ldots A_{n} | X_1 \ldots X_{n}}$ arising from this construction can always be thought of as being generated from a quantum one $\bsig^{\rm{Q}}_{A_1\ldots A_{n} | X_1 \ldots X_{n}}$ by the application of a PTP map $\Lambda^{(B)}[\cdot]$ on the characterised party.

Now, note that 
\begin{align}
\nonumber
p(a_1 \ldots a_n \,b| x_1 \ldots x_n \,y)   &= \Tr{( M^{(1)}_{a_1|x_1} \otimes \ldots \otimes M^{(n)}_{a_{n}|x_{n}} \otimes M^{(B)}_{b|y} ) \mathcal{E}[\rho]}{}\\
\nonumber
&= \Tr{\mathcal{E}^\dagger\left[ M^{(1)}_{a_1|x_1} \otimes \ldots \otimes M^{(n)}_{a_{n}|x_{n}} \otimes M^{(B)}_{b|y}\right] \rho}{}\\
&= \Tr{ M^{(1)}_{a_1|x_1} \otimes \ldots \otimes M^{(n)}_{a_{n}|x_{n}} \otimes {\Lambda^{\dagger(B)}}\left[M^{(B)}_{b|y}\right] \rho}{}
\end{align}
are correlations that have a quantum realisation, for any set of POVMs $\{M_{b|y}^{(B)}\}_{b,y}$ for Bob. This is due to the fact that the dual map $\mathcal{E}^\dagger[\cdot]:=\mathcal{I}^{(1)}\otimes \cdots \otimes \mathcal{I}^{(n)} \otimes \Lambda^{\dagger(B)}[\cdot]$, with $\Lambda^{\dagger(B)}[\cdot]$ the dual of $\Lambda^{(B)}[\cdot]$, factorises into a tensor product of local maps each of which is unital, since $\Lambda^{(B)}[\cdot]$ is trace-preserving. Hence, it maps each tensor product of local POVM elements $M^{(1)}_{a_1|x_1} \otimes \ldots \otimes M^{(n)}_{a_{n}|x_{n}} \otimes M^{(B)}_{b|y}$ to a tensor-product of local POVM elements $M^{(1)}_{a_1|x_1} \otimes \ldots \otimes M^{(n)}_{a_{n}|x_{n}} \otimes \Lambda^{\dagger(B)}\left[M_{b|y}^{(B)}\right]$. 

Thus, assemblages that are constructed in this way can only produce quantum correlations by construction. 
In the following subsection, we discuss how the the properties of $\Lambda^{(B)}[\cdot]$ impart properties onto  $\bsig_{A_1\ldots A_{n} | X_1 \ldots X_{n}}$. 

\subsection{Decomposable PTP maps}
\label{sec:decomposable}

A crucial property of a map for our purposes is the notion of decomposability. A map $\Lambda^{(B)}[\cdot]$ is said to be decomposable whenever it admits a decomposition as $\Lambda^{(B)}[\cdot] = \Lambda_{1}[\cdot] + T \circ \Lambda_{2}[\cdot]$, where $T[\cdot]$ denotes the transposition map~\footnote{Transposition is defined with respect to some chosen local basis; such choice is irrelevant for our purposes as transposition maps in different bases are unitarily related.} and $\Lambda_{1}[\cdot]$ and $\Lambda_{2}[\cdot]$ are CP maps. If $\Lambda^{(B)}$ is trace preserving -- like in the case we are interested in -- then the two CP maps $\Lambda_{1}$ and $\Lambda_2$ form an instrument, that is, $\Lambda_1+\Lambda_2$, besides being obviously completely positive, is also trace preserving. If $\Lambda^{(B)}[\cdot]$ is decomposable, the assemblage it generates via Eq. \eqref{eq:assemchan} is always quantum, no matter which initial quantum assemblage is used, as we are about to prove. 

First, note that the transposition map cannot generate a post-quantum assemblage. This follows from the fact that 
\begin{align*}
&\Tr{ (M^{(1)}_{a_1|x_1} \otimes \ldots \otimes M^{(n)}_{a_{n}|x_{n}} \otimes \id^{(B)}) \rho^{T_B}}{1, \ldots, n} \\
&=\Tr{ (M^{(1)}_{a_1|x_1} \otimes \ldots \otimes M^{(n)}_{a_{n}|x_{n}} \otimes \id^{(B)})^{T_{1\ldots n}} \rho^{T}}{1, \ldots, n} \\
&=\Tr{ (M^{(1)T_1}_{a_1|x_1} \otimes \ldots \otimes M^{(n)T_n}_{a_{n}|x_{n}} \otimes \id^{(B)}) \rho^{T}}{1, \ldots, n}\\
&=\Tr{ (M^{\prime (1)}_{a_1|x_1} \otimes \ldots \otimes M^{\prime (n)}_{a_{n}|x_{n}} \otimes \id^{(B)}) \rho^{\prime}}{1, \ldots, n}\,,
\end{align*} 
where $T_B$ and $T$ denote partial transposition over Bob's subsystem and global transposition over all systems, respectively, $\{M^{\prime (k)}_{a_k|x_k} := M^{(k)T_k}_{a_k|x_k}\}$ are  POVMs, and $\rho^{\prime} := \rho^{T}$ is a quantum state. Hence, the assemblage obtained by local measurements of the steering parties on a partially transposed (on the steered party) quantum state, admits a fully quantum realization.

Now consider a generic decomposable PTP map $\Lambda^{(B)}[\cdot] = \Lambda_{1}[\cdot] + T \circ \Lambda_{2}[\cdot]$, and an arbitrary quantum assemblage $\bsig^{\rm{Q}}_{A_1\ldots A_{n} | X_1 \ldots X_{n}}:=\{\sigma^{\rm{Q}}_{a_1 \ldots a_n | x_1 \ldots x_n}\}$. Then, 
\begin{align}
\label{eq:ass_decomp}
\nonumber
\sigma_{a_1 \ldots a_n | x_1 \ldots x_n} &=\Lambda^{(B)}[\sigma^{\rm{Q}}_{a_1 \ldots a_n | x_1 \ldots x_n}] \\ 
&= \Lambda_{1}[\sigma^{\rm{Q}}_{a_1 \ldots a_n | x_1 \ldots x_n}] + T \circ \Lambda_{2} \nonumber [\sigma^{\rm{Q}}_{a_1 \ldots a_n | x_1 \ldots x_n}] \\
&= p \sigma^{\rm{Q}_1}_{a_1 \ldots a_n | x_1 \ldots x_n} + (1-p)  \sigma^{\rm{Q}_2}_{a_1 \ldots a_n | x_1 \ldots x_n}\,,
\end{align}
where 
\begin{align*}
p &:= \tr{\sum_{a_1, \ldots, a_n} \Lambda_{1}[\sigma^{\rm{Q}}_{a_1 \ldots a_n | x_1 \ldots x_n}]}{}\,,\\
\sigma^{\rm{Q}_1}_{a_1 \ldots a_n | x_1 \ldots x_n} &:= \frac{\Lambda_{1}[\sigma^{\rm{Q}}_{a_1 \ldots a_n | x_1 \ldots x_n}]}{p}\,,\\
\sigma^{\rm{Q}_2}_{a_1 \ldots a_n | x_1 \ldots x_n} &:= \frac{T \circ \Lambda_{2}[\sigma^{\rm{Q}}_{a_1 \ldots a_n | x_1 \ldots x_n}]}{1-p}\,.
\end{align*}
Since $\bsig^{\rm{Q}}_{A_1\ldots A_{n} | X_1 \ldots X_{n}}$ is a quantum assemblage and $\Lambda^{(B)}[\cdot]$ is PTP, $p$ is a valid probability, i.e., $p\in[0,1]$. This, together with the fact that $\Lambda_{1}[\cdot]$ and $\Lambda_{2}[\cdot]$ are CP (trace-non-increasing) maps and that transposition preserves quantum assemblages, implies that both $\bsig^{\rm{Q}_1}_{A_1\ldots A_{n} | X_1 \ldots X_{n}}$ and $\bsig^{\rm{Q}_2}_{A_1\ldots A_{n} | X_1 \ldots X_{n}}$ are quantum assemblages. By convexity of the set of assemblages, it follows then that the assemblage $\bsig_{A_1\ldots A_{n} | X_1 \ldots X_{n}}$ in Eq. \eqref{eq:ass_decomp} is a quantum assemblage too. A direct consequence of this is that no positive PTP maps from qubits to qubits~\footnote{Or from qubits to qutrits, or qutrits to qubits, for that matter.} can generate post-quantum assemblages by the above construction, since all such maps are decomposable~\cite{stormer,woronowicz}.

\subsection{Non-decomposable PTP maps and examples of post-quantum steering}\label{nondecomp}

The observation of Section \ref{sec:decomposable} demonstrates that, if we want to find examples of post-quantum steering by means of the application of positive maps to quantum states, then we must focus on non-decomposable maps.

The question that remains to be answered is whether there exist non-decomposable PTP maps that produce assemblages which are post-quantum. In this section we will provide such an example.

Consider a steering scenario with two uncharacterised parties, who can choose among two dichotomic measurements each. The characterised party will be taken to have a Hilbert space of dimension four. 

We first define a quantum assemblage, assuming that the uncharacterised parties each hold qubits,  i.e. the shared system consists of two qubits and a ququart. The shared state is $\rho = \ket{\Psi}\bra{\Psi}$, where
\begin{align}
\ket{\Psi} = \frac{\ket{\Psi_1}+i\,\ket{\Psi_2}-\ket{\Psi_3}}{\sqrt{14}}\,,
\end{align}
with
\begin{align*}
\ket{\Psi_k} = \sum_{\substack{a_1,a_2,b,b' \in \{0,1\}, \\ a_1+a_2+b+b'= k}} \ket{a_1\, a_2\, b\, b'}\,\, \text{ for } k=1,2,3,
\end{align*}
and where we have introduced the shorthand notation $\ket{a_1\, a_2\, b\, b'} := \ket{a_1}_{A_1} \otimes \ket{a_2}_{A_2} \otimes \ket{b\,b'}_B$. 

The measurements the uncharacterised parties perform on their qubits are: 
\begin{equation}
\begin{aligned}
M_{a_1|0}^{(1)}& =\frac{\id  + (-1)^{a_1} \, X}{2}&
M_{a_1|1}^{(1)}& =\frac{\id + (-1)^{a_1} \,Z}{2}  \\
M_{a_2|0}^{(2)}& =\frac{\id + \tfrac{(-1)^{a_2}}{\sqrt{2}} \, (X+Z)}{2}&
M_{a_2|1}^{(2)} & = \frac{\id + \tfrac{(-1)^{a_2}}{\sqrt{2}} \, (-X+Z)}{2}\,.
\end{aligned}
\end{equation}
where $X$ and $Z$ are Pauli operators. Now define the PTP map $\Lambda^{B}[\cdot]$ as 
\begin{align}
\Lambda^{(B)}[\rho] := \tfrac{1}{2} \, \left( \tr{\rho} \, \id - \rho - U\, \rho^T \, U^\dagger \right)\,,
\end{align} where $U =X \otimes Y$ is an antisymmetric unitary. The ability of the extended reduction criterion to detect states that are positive under partial transposition certifies that $\Lambda^{B}[\cdot]$ is non-decomposable \cite{horodecki2,breuer,hall}. 

The claim now is that $\bsig_{A_1\, A_2 | X_1\,X_2}:=\{\sigma_{a_1a_2|x_1x_2}\}_{a_1,a_2,x_1,x_2}$, with
\begin{align*}
\sigma_{a_1a_2|x_1x_2} := \Lambda^{(B)}\left[\sigma^Q_{a_1a_2|x_1x_2}\right],
\end{align*}
for $\sigma^Q_{a_1a_2|x_1x_2} := \Tr{(M_{a_1|x_1}^{(1)}\otimes M_{a_2|x_2}^{(2)}\otimes \id)\ket{\Psi}\bra{\Psi}}{12}$,  is a post-quantum assemblage. This can be certified numerically via a semidefinite program (SDP). In particular, although the set of quantum assemblages has a complicated structure, it is possible to construct approximations to this set, which have a much simpler structure, and contain within them the set of quantum assemblages~\cite{PQP}.  Whether or not an assemblage is inside such an approximation can be checked efficiently using an SDP, and hence if an assemblage is found to be outside the approximation, then it is also certified to be post-quantum. Using this method, we found that $\bsig_{A_1\, A_2 | X_1\,X_2}$  does not belong to the set of quantum assemblages, and therefore demonstrates post-quantum steering. The details of the calculation can be found in Appendix \ref{sdp}.

We emphasise that this is the first analytical example of a post-quantum assemblage that can only produce quantum correlations in a Bell experiment where the characterised party makes measurements. Although we will not discuss the details of this, we have verified in a similar fashion that also the well-known Choi map~\cite{choi1,choi2} can generate post-quantum assemblages.

\section{Post-quantum steering from unextendible product bases}
\label{sec:UPB}

In this section we present a family of certifiable post-quantum assemblages for arbitrary multipartite steering scenarios, which arises naturally from our formalism. We will consider the more general scenario, where instead of a single characterised party, we have $t$ characterised parties, who are steered by $n$ uncharacterised parties performing $m$ measurements of $d$ outcomes. 

We take a local-orthogonality (LO) inequality \cite{LO2} in the $(n+t,m,d)$ Bell scenario. Following \cite{LO2}, one can find an unextendible product basis (UPB) or a weak UPB (for scenarios with nondichotomic measurements) for $\mathcal{H} = \left(\mathbb{C}^d\right)^{\otimes (n+t)}$ from the LO inequality. Such a weak UPB can be constructed as follows \cite{LO2}. In each local Hilbert space $\mathbb{C}^{d}$, we distinguish $m$ different orthogonal bases, denoted by $B_j=\{\ket{\phi_i^{(j)}}\}_{i=0}^{d-1}$, where
$j=0,\ldots,m-1.$\footnote{For simplicity, we take these to be the same for
all sites.} These bases are chosen such that if two basis vectors are orthogonal, then they are from the same basis: $ \langle \phi_i^{(j)} | \phi_{i'}^{(j')}\rangle = 0 \,\Longrightarrow\, j = j'$. Given an optimal LO inequality represented by a set of mutually orthogonal events $\mathcal{S}$, the corresponding UPB consists of the following elements: $\left\{ \ket{\phi_{a_1}^{(x_1)}}\otimes \ldots \otimes 
\ket{\phi_{a_{n+t}}^{(x_{n+t})}} | (a_1 \ldots a_{n+t} |x_1 \ldots x_{n+t}) \in \mathcal{S}\right\}$. 

This UPB then defines a normalised entanglement witness $W = f(\epsilon) \, \left( \Pi_{\mathrm{UPB}} - \epsilon \mathbbm{1} \right)$, where $\epsilon = \min_{\ket{\psi_1} \otimes \ldots \otimes \ket{\psi_{n+t}}} \bra{\psi_1} \otimes \ldots \otimes \bra{\psi_{n+t}} \Pi_{\mathrm{UPB}} \ket{\psi_1} \otimes \ldots \otimes \ket{\psi_{n+t}}$, and 
$f(\epsilon)=(|\mathcal{S}|-d^{n+t} \,\epsilon)^{-1}$. 
Indeed, since $\epsilon \in \left(0,\tfrac{|\mathcal{S}|}{d^{n+t}}\right)$, $\tr{\rho\,W}$ gives nonnegative values when $\rho$ is a fully separable state, and $\tr{\rho_{\sf{be}}\,W}<0$ for the bound entangled state $\rho_{\mathsf{be}} := \tfrac{1}{d^{n+t}-|\mathcal{S}|}\left( \mathbbm{1} - \Pi_{\mathrm{UPB}} \right)$.
The method of Example \ref{ex:wit} can then be applied to this weak UPB to construct an assemblage. This is defined by the uncharacterised parties performing the measurements $M^{(j)}_{a_j|x_j}:=\ket{\phi_{a_j}^{(x_j)}}\bra{\phi_{a_j}^{(x_j)}}$, $j=1\ldots n$, on $W$:
\begin{align*}
\sigma_{a_1 \ldots a_{n} | x_1 \ldots x_{n}} = \Tr{(\ket{\phi_{a_1}^{(x_1)}}\bra{\phi_{a_1}^{(x_1)}}\otimes \ldots \otimes 
\ket{\phi_{a_{n}}^{(x_{n})}}\bra{\phi_{a_{n}}^{(x_{n})}} \otimes \mathbbm{1}^{\otimes t}) \, W}{1\ldots n}\,.
\end{align*}

The post-quantumness of the assemblage is certified by the correlations obtained when the characterised parties measure $M^{(j)}_{a_j|x_j}:=\ket{\phi_{a_{j}}^{(x_{j})}}\bra{\phi_{a_{j}}^{(x_{j})}}$, $j=n+1 \ldots n+t$, that is:
\begin{align*}
p(a_1 \ldots a_{n+t} | x_1 \ldots x_{n+t}) = \tr{(\ket{\phi_{a_1}^{(x_1)}}\bra{\phi_{a_1}^{(x_1)}}\otimes \ldots \otimes 
\ket{\phi_{a_{n+t}}^{(x_{n+t})}}\bra{\phi_{a_{n+t}}^{(x_{n+t})}}) W}\,.
\end{align*}
Indeed, these correlations violate the original LO inequality
\begin{align*}
\sum_{(a_1 \ldots a_{n+t} | x_1 \ldots x_{n+t}) \in \mathcal{S}} p(a_1 \ldots a_{n+t} | x_1 \ldots x_{n+t}) \leq 1\,,
\end{align*}
since 
\begin{align*}
\sum_{(a_1 \ldots a_{n+t} | x_1 \ldots x_{n+t}) \in \mathcal{S}} \tr{\ket{\phi_{a_1}^{(x_1)}}\bra{\phi_{a_1}^{(x_1)}}\otimes \ldots \otimes 
\ket{\phi_{a_{n+t}}^{(x_{n+t})}}\bra{\phi_{a_{n+t}}^{(x_{n+t})}} \, W} = f(\epsilon) \, |\mathcal{S}|\,(1-\epsilon)\,,
\end{align*}
which is larger than unity since $\epsilon \in (0,\tfrac{|\mathcal{S}|}{d^{n+t}})$.

Even though the post-quantum assemblages that arise in this family produce post-quantum correlations, the fact that they admit such an elegant analytical form makes them interesting, as this may be useful for potential applications.

\section{A post-quantum steering quantifier}
\label{sec:quantifier}

A crucial issue in the theory of steering is its quantification -- i.e. a notion of whether one assemblage demonstrates more steering than another in some well-defined sense. A number of quantifiers have recently been explored \cite{SNC14,piani,PQP,GA15,KW16}, arising from differing operational tasks or geometrical constructions.

The formalism presented in Sec.~\ref{se:for} naturally leads to a novel steering quantifier, similar in spirit to that proposed in Ref.~\cite{piani2} for Bell correlations, which we refer to as the \textit{steering negativity}. The steering negativity is specially tailored to quantify the amount post-quantum steering an assemblage demonstrates (as opposed to the amount of steering), as we see next.

By virtue of Theorem \ref{theOteo}, any assemblage can be reproduced by local quantum measurements on a Hermitian operator $\tilde{\rho}$. This operator, which is not unique, can always be decomposed in terms of its negative and positive parts, i.e. $\tilde{\rho} = \rho_+ - \rho_-$, with $\rho_{\pm} \geq 0$. Then, for an arbitrary no-signalling assemblage $\bsig_{A_1\ldots A_{n} | X_1 \ldots X_{n}}$, we define its steering negativity as
\begin{align}
\label{eq:def_neg}
\nu(\bsig_{A_1\ldots A_{n} | X_1 \ldots X_{n}}) := \min_{\big\{M^{(i)}_{x_i}\big\},\,\tilde{\rho}} &\quad \tr{\rho_-} \\
\nonumber
\text{s.t.} &\quad \tilde{\rho} = \rho_+ - \rho_-, \nonumber \\
&\quad \rho_\pm \geq 0, \nonumber \\
&\quad \sigma_{a_1 \ldots a_{n} | x_1 \ldots x_{n}} = \Tr{ \left( M^{(1)}_{a_1|x_1} \otimes \ldots \otimes M^{(n)}_{a_{n}|x_{n}} \otimes \mathbbm{1} \right) \, \tilde{\rho}}{1, \ldots n} \nonumber 
\end{align}
where $M^{(i)}_{x_i}$ stands for a POVM with elements $M^{(i)}_{a_i|x_i}$, and the minimisation runs over all such $M^{(i)}_{x_i}$, for $1\leq i\leq n$, as well as over $\tilde{\rho}$. 
Note that since all quantum assemblages admit a decomposition as in Eq. \eqref{Ogen} with a positive semidefinite $\tilde{\rho}$, their steering negativity by definition, is zero. Hence, in contrast to other measures of steering, this figure of merit is relevant for quantifying the post-quantumness of an assemblage.

Since the operator $\tilde{\rho}$ is normalised, the negativity can equivalently be computed as 
\begin{align}
\nonumber
\nu(\bsig_{A_1\ldots A_{n} | X_1 \ldots X_{n}})  \equiv  \min_{\big\{M^{(i)}_{x_i}\big\},\,\tilde{\rho}}& \quad  \frac{\|\tilde{\rho}\|_1 - 1}{2} \\
\text{s.t.} &\quad \sigma_{a_1 \ldots a_{n} | x_1 \ldots x_{n}} = \Tr{ \left( M^{(1)}_{a_1|x_1} \otimes \ldots \otimes M^{(n)}_{a_{n}|x_{n}} \otimes \mathbbm{1} \right) \, \tilde{\rho}}{1, \ldots n} \nonumber 
\label{eq:neg_alt_exp}
\end{align}
where $\|\cdot\|_1$ denotes the trace norm. This alternative expression for  $\nu$ makes the connection with the well-known negativity \cite{ZHSL98,VW02} from entanglement theory explicit. In fact, if $\tilde{\rho}$  is taken as the partial transpose of a given state $\rho$, then 
$\frac{||\tilde{\rho}||_1 - 1}{2}$ defines precisely the entanglement negativity of $\rho$. 

In the following, we will show that the steering negativity is a convex quantifier of post-quantum steering. We do so by putting forward a proposal for the study post-quantum steering from a resource-theoretic perspective, whereby Alice and Bob are allowed to perform operations which are deemed unable to increase the amount of post-quantum steering they share (so called free operations), similar to what has been done for (quantum) steering \cite{GA15}.

The quantum steering exhibited by a quantum assemblage is postulated in Ref.~\cite{GA15} not to increase under local operations and one-way classical communication (one-way LOCC), where the communication is only allowed from the steered party to the steering parties. On the other hand, a shared quantum state cannot lead to post-quantum steering, by definition. Led by the idea of combining these two properties, we postulate that post-quantum steering should not increase under the processing of an assemblage via the introduction of a further shared entangled state and one-way LOCC. We notice that, thanks to quantum teleportation~\cite{teleportation}, unrestricted shared entanglement assisted by one-way LOCC is equivalent to local operations aided by one-way quantum communication. Thus, much in the same fashion in which quantum steering is postulated not to increase under one-way LOCC, our request amounts to postulating that post-quantum steering does not increase under one-way quantum operations, with the communication going from the steered party to the steering parties. Notice that, since classical communication is a subset of quantum communication, a post-quantum steering quantifier that respects our request is necessarily also a standard steering monotone. Moreover, given that one-way quantum communication allows for the sharing of an arbitrary quantum state, and hence for the creation -- even from scratch -- of an arbitrary quantum assemblage, a post-quantum steering quantifier necessarily assumes a constant value for all quantum assemblages, and such a value can be set to zero. What we exactly mean by processing of an assemblage by one-way quantum operations is shown in Figure~\ref{fig:processing}, where for the sake of simplicity and clarity we depict explicitly only one steering party.

\begin{figure}
	\centering
	\subfigure[$\tilde{\rho}$-operator view of an assemblage.]{
		\begin{tikzpicture}[scale=1.000000,x=1pt,y=1pt]
		\filldraw[color=white] (0.000000, -7.500000) rectangle (37.000000, 37.500000);
		\draw[color=black] (0.000000,29.500000) -- (37.000000,29.500000);
		\draw[color=black] (0.000000,30.500000) -- (37.000000,30.500000);
		\draw[color=black] (0.000000,30.000000) node[left] {$x$};
		\draw[color=black] (0.000000,15.000000) -- (18.500000,15.000000);
		\draw[color=black] (0.000000,0.000000) -- (37.000000,0.000000);
		\filldraw[color=white,fill=white] (0.000000,-3.750000) rectangle (-4.000000,18.750000);
		\draw[decorate,decoration={brace,amplitude = 4.000000pt},very thick] (0.000000,-3.750000) -- (0.000000,18.750000);
		\draw[color=black] (-4.000000,7.500000) node[left] {$\tilde{\rho}_{AB}$};
		\draw (18.000000,30.000000) -- (18.000000,15.000000);
		\draw (19.000000,30.000000) -- (19.000000,15.000000);
		\begin{scope}
		\draw[fill=white] (18.500000, 22.500000) +(-45.000000:17.677670pt and 19.091883pt) -- +(45.000000:17.677670pt and 19.091883pt) -- +(135.000000:17.677670pt and 19.091883pt) -- +(225.000000:17.677670pt and 19.091883pt) -- cycle;
		\clip (18.500000, 22.500000) +(-45.000000:17.677670pt and 19.091883pt) -- +(45.000000:17.677670pt and 19.091883pt) -- +(135.000000:17.677670pt and 19.091883pt) -- +(225.000000:17.677670pt and 19.091883pt) -- cycle;
		\draw (18.500000, 22.500000) node {$M^A_{a|x}$};
		\end{scope}
		\draw[color=black] (37.000000,30.000000) node[right] {$a$};
		\draw[color=black] (37.000000,0.000000) node[right] {$\sigma^B_{a|x}$};
		\end{tikzpicture}
	}
	\qquad
	\subfigure[Processing of an assemblage by one-way quantum operations.]{
		\begin{tikzpicture}[scale=1.000000,x=1pt,y=1pt]
		\filldraw[color=white] (0.000000, -7.500000) rectangle (224.000000, 67.500000);
		\draw[color=black] (0.000000,44.500000) -- (62.000000,44.500000);
		\draw[color=black] (0.000000,45.500000) -- (62.000000,45.500000);
		\draw[color=black] (62.000000,44.500000) -- (69.500000,44.500000);
		\draw[color=black] (62.000000,45.500000) -- (69.500000,45.500000);
		\draw[color=black] (69.500000,44.500000) -- (77.000000,44.500000);
		\draw[color=black] (69.500000,45.500000) -- (77.000000,45.500000);
		\draw[color=black] (77.000000,44.500000) -- (106.500000,44.500000);
		\draw[color=black] (77.000000,45.500000) -- (106.500000,45.500000);
		\draw[color=black,thick] (106.500000,44.500000) -- (200.500000,44.500000);
		\draw[color=black,thick] (106.500000,45.500000) -- (200.500000,45.500000);
		\draw[color=black,thin] (200.500000,44.500000) -- (224.000000,44.500000);
		\draw[color=black,thin] (200.500000,45.500000) -- (224.000000,45.500000);
		\draw[color=black] (0.000000,45.000000) node[left] {$x'$};
		\draw[color=black,thick] (0.000000,30.000000) -- (62.000000,30.000000);
		\draw[color=black,thick] (62.000000,30.000000) -- (69.500000,30.000000);
		\draw[color=black,thick] (69.500000,30.000000) -- (77.000000,30.000000);
		\draw[color=black,thick] (77.000000,30.000000) -- (153.500000,30.000000);
		\draw[color=black,thick] (0.000000,15.000000) -- (28.000000,15.000000);
		\draw[color=black,thin] (28.000000,15.000000) -- (62.000000,15.000000);
		\draw[color=black,thin] (62.000000,15.000000) -- (69.500000,15.000000);
		\draw[color=black,thin] (69.500000,15.000000) -- (77.000000,15.000000);
		\draw[color=black,thin] (77.000000,15.000000) -- (224.000000,15.000000);
		\filldraw[color=white,fill=white] (0.000000,11.250000) rectangle (-4.000000,33.750000);
		\draw[decorate,decoration={brace,amplitude = 4.000000pt},very thick] (0.000000,11.250000) -- (0.000000,33.750000);
		\draw[color=black] (-4.000000,22.500000) node[left] {$\tilde{\rho}_{AB}$};
		\draw[color=black,rounded corners=4.000000pt] (28.000000,0.000000) -- (62.000000,0.000000) -- (69.500000,30.000000);
		\draw[color=black,rounded corners=4.000000pt] (69.500000,30.000000) -- (77.000000,60.000000) -- (200.500000,60.000000);
		\draw (28.000000,15.000000) -- (28.000000,0.000000);
		\begin{scope}
		\draw[fill=white] (28.000000, 7.500000) +(-45.000000:31.112698pt and 19.091883pt) -- +(45.000000:31.112698pt and 19.091883pt) -- +(135.000000:31.112698pt and 19.091883pt) -- +(225.000000:31.112698pt and 19.091883pt) -- cycle;
		\clip (28.000000, 7.500000) +(-45.000000:31.112698pt and 19.091883pt) -- +(45.000000:31.112698pt and 19.091883pt) -- +(135.000000:31.112698pt and 19.091883pt) -- +(225.000000:31.112698pt and 19.091883pt) -- cycle;
		\draw (28.000000, 7.500000) node {$\Lambda_{B\rightarrow B'A'}$};
		\end{scope}
		\draw (106.000000,60.000000) -- (106.000000,45.000000);
		\draw (107.000000,60.000000) -- (107.000000,45.000000);
		\begin{scope}
		\draw[fill=white] (106.500000, 52.500000) +(-45.000000:24.748737pt and 19.091883pt) -- +(45.000000:24.748737pt and 19.091883pt) -- +(135.000000:24.748737pt and 19.091883pt) -- +(225.000000:24.748737pt and 19.091883pt) -- cycle;
		\clip (106.500000, 52.500000) +(-45.000000:24.748737pt and 19.091883pt) -- +(45.000000:24.748737pt and 19.091883pt) -- +(135.000000:24.748737pt and 19.091883pt) -- +(225.000000:24.748737pt and 19.091883pt) -- cycle;
		\draw (106.500000, 52.500000) node {$\mathcal{E}^{A'}_{x|x'}$};
		\end{scope}
		\draw[thick] (153.000000,45.000000) -- (153.000000,30.000000);
		\draw[thick] (154.000000,45.000000) -- (154.000000,30.000000);
		\begin{scope}[thick]
		\begin{scope}
		\draw[fill=white] (153.500000, 37.500000) +(-45.000000:24.748737pt and 19.091883pt) -- +(45.000000:24.748737pt and 19.091883pt) -- +(135.000000:24.748737pt and 19.091883pt) -- +(225.000000:24.748737pt and 19.091883pt) -- cycle;
		\clip (153.500000, 37.500000) +(-45.000000:24.748737pt and 19.091883pt) -- +(45.000000:24.748737pt and 19.091883pt) -- +(135.000000:24.748737pt and 19.091883pt) -- +(225.000000:24.748737pt and 19.091883pt) -- cycle;
		\draw (153.500000, 37.500000) node {$M^A_{a|x}$};
		\end{scope}
		\end{scope}
		\draw (200.000000,60.000000) -- (200.000000,45.000000);
		\draw (201.000000,60.000000) -- (201.000000,45.000000);
		\begin{scope}
		\draw[fill=white] (200.500000, 52.500000) +(-45.000000:24.748737pt and 19.091883pt) -- +(45.000000:24.748737pt and 19.091883pt) -- +(135.000000:24.748737pt and 19.091883pt) -- +(225.000000:24.748737pt and 19.091883pt) -- cycle;
		\clip (200.500000, 52.500000) +(-45.000000:24.748737pt and 19.091883pt) -- +(45.000000:24.748737pt and 19.091883pt) -- +(135.000000:24.748737pt and 19.091883pt) -- +(225.000000:24.748737pt and 19.091883pt) -- cycle;
		\draw (200.500000, 52.500000) node {$N^{A'}_{a'|a}$};
		\end{scope}
		\draw[color=black] (224.000000,45.000000) node[right] {$a'$};
		\draw[color=black] (224.000000,15.000000) node[right] {$\sigma^{B'}_{a'|x'}$};
		\draw[draw opacity=1.000000,fill opacity=0.200000,color=blue,fill=blue] (86.000000,67.500000) rectangle (221.000000,22.500000);
		\draw (153.500000, 67.500000) node[text width=144pt,above,text centered,color=blue] {$M'^{AA'}_{a'|x'}$};
		\draw[draw opacity=1.000000,fill opacity=0.200000,color=blue,fill=blue] (86.000000,67.500000) rectangle (221.000000,22.500000);
		\draw[draw opacity=1.000000,fill opacity=0.200000,color=red,fill=red] (3.000000,37.500000) rectangle (80.000000,-7.500000);
		\draw (41.500000, -7.500000) node[text width=144pt,below,text centered,color=red] {$\tilde{\rho}'_{AA'B'}$};
		\draw[draw opacity=1.000000,fill opacity=0.200000,color=red,fill=red] (3.000000,37.500000) rectangle (80.000000,-7.500000);
		\end{tikzpicture}
	}
	\caption{Circuit representation of an assemblage, 
		and processing of an assemblage by means of one-way quantum operations. Time goes from left to right. For the sake of clarity we focus on the case of just one steering party. (a) A no-signalling assemblage $\{\sigma^B_{a | x}\}_{a,x}$ is seen as the result of local measurements, described by a set of POVMs $\{M_{a|x}\}_{a,x}$, performed by a steering party on part of a distributed system, which is initially in a (pseudo-)state $\tilde{\rho}_{AB}$ (see Theorem \ref{Ogen}). (b) The original assemblage (whose $\tilde{\rho}$-representation is depicted in bold in the diagram) can be processed by one-way quantum operations. The steered party applies a channel $\Lambda_{B\rightarrow B'A'}$ on their share of the system, and sends system $A'$ to the steering party. Based on a classical input $x'$, the steering party applies an instrument $\{\mathcal{E}_{x|x'}\}_x$ to the received system $A'$. The classical output of the instrument, $x$, is used as input for the original steering scenario, that is, as choice of original measurement on the $A$ part of the originally shared (pseudo-)state $\tilde{\rho}_{AB}$, while the quantum output of the instrument is kept in a quantum memory for further processing (notice that such quantum output may include information about both $x$ and $x'$, without loss of generality). The output $a$ of the original measurement is used to decide which final POVM $\{N_{a'|a}\}_a$ to implement on the local quantum memory, with final classical output $a'$. The processing can be described in terms of a new (pseudo-)state $\tilde{\rho}'_{AA'B}=\Lambda_{B\rightarrow B'A'}[\tilde{\rho}_{AB}]$ (highlighted in red online) and a new set of POVMs of the steering party, $\{M'^{AA'}_{a'|x'}\}_{a',x'}$ (highlighted in blue online). The end result is a new assemblage $\{\sigma^{B'}_{a' | x'}\}_{a',x'}$. We remark that the processing and the final assemblage are independent from the specific representation of the original assemblage, and depend only on the original assemblage, as well as on the choices of channel $\Lambda_{B\rightarrow B'A'}$, instruments $\{\mathcal{E}_{x|x'}\}_{x,x'}$, and POVMs $\{N_{a'|a}\}_{a,a'}$.}
	\label{fig:processing}
\end{figure}
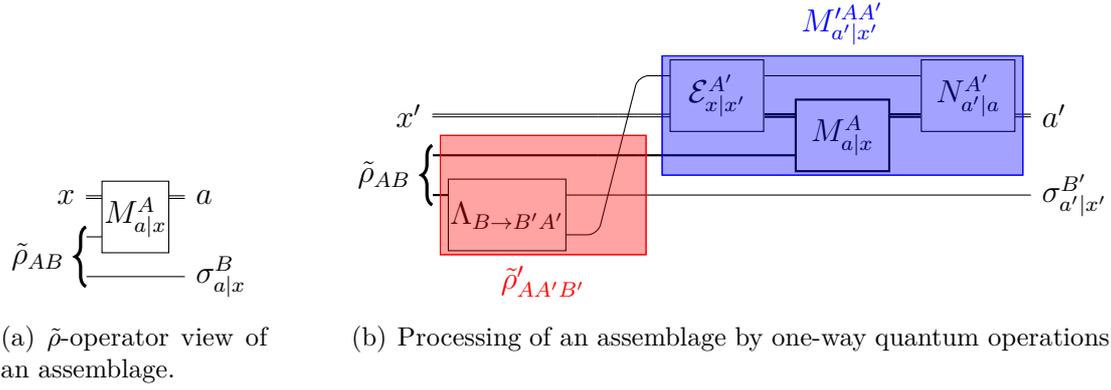

As with many quantum resource theories, it is also convenient and reasonable -- although not strictly necessary~\cite{plenionegativity} -- to ask that a post-quantum steering quantifier is convex. 

We will see below that the steering negativity is a valid convex post-quantum steering quantifier, in the sense that it respects the requests delineated above.

\begin{thm}[Convexity of $\nu$]
\label{the:convex_mon}
The steering negativity is a convex steering quantifier. That is, it is non-increasing under 
arbitrary convex mixings, 
\begin{equation}
\label{eq:conv_nu}
\nu\left(q\,\bsig+(1-q)\bsig'\right)\leq q\,\nu(\bsig)+(1-q)\,\nu(\bsig'), \text{ for all } \bsig \text{ and } \bsig', \text{ and all } 0\leq q\leq 1.
\end{equation}
\end{thm}

\begin{proof}
Let $\tilde{\rho}=\rho_+-\rho_-$ and $\tilde{\rho}'=\rho_+'-\rho_-'$ be optimal Hermitian operators  attaining the minima in Eq. \eqref{eq:def_neg} for the assemblages $\bsig_{A_1\ldots A_{n} | X_1 \ldots X_{n}}$ and $\bsig_{A_1\ldots A_{n} | X_1 \ldots X_{n}}'$, respectively, for two suitable sets of POVMs $\big\{M^{(i)}_{x_i}\big\}:=\big\{M^{(1)}_{a_1|x_1}, \hdots, M^{(n)}_{a_n|x_n}\big\}_{a_1,x_1,\hdots a_n,x_n}$ and $\big\{M'^{(i)}_{x_i}\big\}:=\big\{M'^{(1)}_{a_1|x_1}, \hdots, M'^{(n)}_{a_n|x_n}\big\}_{a_1,x_1,\hdots a_n,x_n}$. This implies that $\nu(\bsig_{A_1\ldots A_{n} | X_1 \ldots X_{n}})=\tr{\rho_-}$ and $\nu(\bsig_{A_1\ldots A_{n} | X_1 \ldots X_{n}}')=\tr{\rho_-'}$. Now, consider the state 
\begin{align}
\label{eq:conv_comb_trick}
\nonumber
\tilde{\rho}_{A^*_1,A_1,\hdots A^*_n,A_n,B}&:=q\, \ket{0}\bra{0}_{A^*_1}\otimes\hdots \ket{0}\bra{0}_{A^*_n}\otimes\tilde{\rho}+(1-q)\,\ket{1}\bra{1}_{A^*_1}\otimes\hdots \ket{1}\bra{1}_{A^*_n}\otimes\tilde{\rho}'\\
\nonumber
&=\left(q\, \ket{0}\bra{0}_{A_1^*}\otimes\hdots \ket{0}\bra{0}_{A_n^*}\otimes \rho_++(1-q)\,\ket{1}\bra{1}_{A_1^*}\otimes\hdots \ket{1}\bra{1}_{A_n^*}\otimes \rho_+'\right)\\
&-\left(q\, \ket{0}\bra{0}_{A_1^*}\otimes\hdots \ket{0}\bra{0}_{A_n^*}\otimes \rho_-+(1-q)\,\ket{1}\bra{1}_{A_1^*}\otimes\hdots \ket{1}\bra{1}_{A_n^*}\otimes \rho_-'\right),
\end{align}
where a local ancillary qubit  $A^*_i$, in state either $\ket{0}_{A^*_i}$ or $\ket{1}_{A^*_i}$, has been given to each Alice, with $1\leq i\leq n$. This state realises a decomposition of the form Eq. \eqref{Ogen} for $q\,\bsig_{A_1\ldots A_{n} | X_1 \ldots X_{n}}+(1-q)\,\bsig_{A_1\ldots A_{n} | X_1 \ldots X_{n}}'$, where a suitable set of POVMs can be taken to be $\big\{q\, \ket{0}\bra{0}_{A_1^*}\otimes M^{(1)}_{a_1|x_1}+(1-q)\,\ket{1}\bra{1}_{A_1^*}\otimes M'^{(1)}_{a_1|x_1},\, \hdots\, q\, \ket{0}\bra{0}_{A_n^*}\otimes M^{(n)}_{a_n|x_n}+(1-q)\,\ket{1}\bra{1}_{A_n^*}\otimes M'^{(n)}_{a_n|x_n}\big\}_{a_1,x_1,\hdots a_n,x_n}$. Therefore, even though such a decomposition is not guaranteed to be optimal, it is nevertheless the case that
\begin{align}
\nonumber
\nu\left(q\,\bsig_{A_1\ldots A_{n} | X_1 \ldots X_{n}}+(1-q)\bsig_{A_1\ldots A_{n} | X_1 \ldots X_{n}}'\right)&\leq{\rm tr}\big\{q\, \ket{0}\bra{0}_{A_1'}\otimes\hdots \ket{0}\bra{0}_{A_n'}\otimes \rho_-\\
\nonumber
&+(1-q)\,\ket{1}\bra{1}_{A_1'}\otimes\hdots \ket{1}\bra{1}_{A_n'}\otimes \rho_-'\big\}\\
&=q\,\tr{\rho_-}+(1-q)\,\tr{\rho_-'}.
\label{eq:last_conv}
\end{align}
Note that the last term  equals the right hand of Eq. \eqref{eq:conv_nu}, which proves the theorem's statement. 
\end{proof}

\begin{thm}[Monotonicity of $\nu$]
The steering negativity $\nu$ is a post-quantum steering monotone under processing by one-way quantum operations. 
\end{thm}
\begin{proof}
	Let the pseudo-state $\tilde{\rho}_{A_1A_2\ldots A_n B}$ be optimal for the sake of computing the steering negativity of a given steering assemblage. Figure~\ref{fig:processing}.(a) shows how processing such assemblage by one-way quantum operations from the steered party to the steering parties leads to a new assemblage that may be thought as originating from a shared (pseudo-)state
	\[
		\tilde{\rho}'_{A_1A'_1A_2A'_2\ldots A_n A'_n B'} = \Lambda_{B\rightarrow A'_1A'_2\ldots A'_n B'}[\tilde{\rho}_{A_1A_2\ldots A_n B}] ,
	\]
	where $\Lambda_{B\rightarrow A'_1A'_2\ldots A'_n B'}$ is a completely-positive trace-preserving map. While such an $\tilde{\rho}'$ may not be optimal for the sake of the steering negativity of the new assemblage, since the trace norm does not increase under the partial action of a completely positive and trace-preserving map, this is enough to prove that the steering negativity does not increase under processing by one-way quantum operations.
\end{proof}

\section{Discussion}

The scope of the steering phenomenon has been widely studied with respect to its applications, for instance to engineer one-sided device independent information theoretical protocols robust to loopholes \cite{sqkd,sr1,sr2,smi1, smi2, smi3,sst1, sst2}. However, questions about its implication for our fundamental understanding of Nature have been much less addressed. In this work we developed a framework that allows us to understand steering in more general set-ups and potentially in theories beyond quantum mechanics. Our formalism starts from the usual formulation of a quantum steering experiment, where the uncharacterised parties perform measurements on their share of a system. By relaxing the properties of the mathematical object $\tilde{\rho}$ that represents the state of the system, one can simulate steering experiments beyond what quantum mechanics allows, while still complying with physical assumptions such as no-signalling.  This framework provides a way to understand classical, quantum and post-quantum steering in a unified manner, each of which can be recovered as special cases of the formalism. In particular, our approach comes equipped with an inherent functional that quantifies the post-quantumness of an assemblage, the negativity of post-quantum steering. We postulate that post-quantum steering should not increase under one-way quantum operations from the steered parties to the steering parties, whereas standard quantum steering is postulated not to increase under one-way LOCC~\cite{GA15}. We prove that the negativity of post-quantum steering is respects such a postulate, and more precisely that it is a convex post-quantum-steering monotone. 

By exploring the connections between entanglement witnesses and positive but not completely positive maps, our framework especially succeeds in representing post-quantum assemblages that may only generate quantum correlations. Using this method, we were able to generate the first analytical examples of post-quantum assemblages which cannot exhibit post-quantum Bell non-locality. An open question is whether every non-decomposable positive map can produce post-quantum assemblages given a suitable initial quantum steering experiment (i.e. local measurements on a quantum state). Along these lines lies the question of what type of entanglement properties should the state $\rho$ of the system shared by all the parties have such that, when the steered one applies a non-decomposable positive map to their quantum system, the generated assemblage is post-quantum. More broadly, our formalism also allows for the definition of Gleason assemblages, which generalise quantum ones. We  provided a family of entanglement witnesses and measurements, constructed from unextendible product bases and local orthogonality inequalities, such that the Gleason assemblages they generate are provable post-quantum. 

Although, post-quantum non-locality and post-quantum steering are fundamentally distinct concepts, there are still many opportunities to explore their relationship. For example, if we take a post-quantum assemblage that can never exhibit post-quantum non-locality, is it possible to take multiple copies of this assemblage and apply some filtering process to reveal post-quantum non-locality? We dub this concept \textit{hidden post-quantum non-locality}, and it remains open whether this can occur and, furthermore, whether it might be the case that in fact all post-quantum assemblages exhibit it.

It would also be fascinating to try and find tasks for which post-quantum steering gives a clear advantage over standard quantum steering. One candidate task is entanglement-assisted sub-channel discrimination with one-way measurements \cite{piani}, where it is known that it is steering, rather than simple entanglement~\cite{pianiwatrous}, that gives an advantage. Post-quantum steering might also help trivialise certain communication tasks (cf. Ref. \cite{Belen16}). We leave it for future work whether post-quantum steering is more useful for any of these tasks, and whether the formalism introduced here might facilitate the study of this question. 

In conclusion, these analytical formulations of post-quantum assemblages provide a starting point from where to explore the possible physical or information-theoretical consequences that the phenomenon could have. We believe that such an approach may shed light on the problem of characterising quantum steering from basic physical principles and of understanding the possibilities and limitations of the steering phenomenon in Nature. 

\section*{Acknowledgments}

We thank Nicolas Brunner and Rodrigo Gallego for fruitful discussions. 
ABS and PS acknowledge financial support from ERC AdG NLST. MJH acknowledges funding from the EPSRC through the Networked Quantum Information Technologies (NQIT) Hub.
LA acknowledges the Brazilian ministries MEC and MCTIC and agencies CNPq, CAPES, FAPERJ, and INCT-IQ, for financial support, and the International Institute of Physics (IIP) at Natal for the hospitality and financial support. PS acknowledges support from the Royal Society, through the University Research Fellowship UHQT. MP acknowledges support from European Union's Horizon 2020 Research and Innovation Programme under the Marie Sk{\l}odowska-Curie Action OPERACQC (Grant Agreement No. 661338), and from the Foundational Questions Institute under the Physics of the Observer Programme (Grant No. FQXi-RFP-1601). This research was supported in part by Perimeter Institute for Theoretical Physics. Research at Perimeter Institute is supported by the Government of Canada through the Department of Innovation, Science and Economic Development Canada and by the Province of Ontario through the Ministry of Research, Innovation and Science.

\appendix

\section{Certification of a post-quantum assemblage} \label{sdp}

In Sec. \ref{nondecomp} an assemblage is generated from a non-decomposable positive but not completely positive map being applied to the characterised part of a quantum assemblage. For completeness we reproduce the recipe for this assemblage. Each Alice can choose among two dichotomic measurements. The characterised party, Bob, describes the marginal state of his system by a Hilbert space of dimension 4. The quantum assemblage under consideration arises as follows: the shared state the uncharacterised parties measure on is given by
\begin{align}
\ket{\Psi} = \frac{\ket{\Psi_\mathsf{1}}+i\,\ket{\Psi_\mathsf{2}}-\ket{\Psi_\mathsf{3}}}{\sqrt{14}}\,,
\end{align}
where
\begin{align*}
\ket{\Psi_\mathsf{N}} = \sum_{\substack{a_1,a_2,b,b' \in \{0,1\}, \\ a_1+a_2+b+b' == N}} \ket{a_1\, a_2 \, b \, b'}\,.
\end{align*}
Here $\ket{a_1\, a_2 \, b \, b'}$ is understood as $\ket{a_1\, a_2 \, b \, b'} = \ket{a_1}_{A_1} \otimes \ket{a_2}_{A_2} \otimes \ket{bb'}_B$. 

The projective measurements the uncharacterised parties perform on their qubits are: 
\begin{align*}
M_{a_1|0}^{(1)}& =\frac{\id  + (-1)^{a_1} \, X}{2}&
M_{a_1|1}^{(1)}& =\frac{\id + (-1)^{a_1} \,Z}{2}  \\
M_{a_2|0}^{(2)}& =\frac{\id + \tfrac{(-1)^{a_2}}{\sqrt{2}} \, (X+Z)}{2}&
M_{a_2|1}^{(2)} & = \frac{\id + \tfrac{(-1)^{a_2}}{\sqrt{2}} \, (-X+Z)}{2}\,.
\end{align*}
where $X$ and $Z$ are the Pauli matrices. The quantum assemblage then arises via
\begin{align}
\sigma^Q_{a_1a_2|x_1x_2}=\Tr{M_{a_1|x_1}^{(1)} \otimes M_{a_2|x_2}^{(2)} \otimes \id_4 \, \ket{\Psi}\bra{\Psi}}{A_1A_2}\,.
\end{align}
Now define the positive-trace-preserving map $\Lambda^{(B)} \,:\, \mathcal{H}_B \rightarrow \mathcal{H}_B$ as 
\begin{align}
\Lambda^{(B)}[\rho] := \tfrac{1}{2} \, \left( \tr{\rho} \, \id_4 - \rho - U\, \rho^T \, U^\dagger \right)\,,
\end{align}
where $U = X \otimes Y$.

We now need to check whether the assemblage is almost-quantum. Checking whether an assemblage is almost-quantum is an SDP. As a bi-product of the computation, we obtain a steering inequality that also certifies the post-quantumness of $\{\sigma_{a_1a_2|x_1x_2}\}$. The inequality has the form
\begin{align}\label{eq:steeringineq}
\Tr{F_R \, \tilde{\rho}_B + \sum_{x_1} F^{(1)}_{x_1} \, \tilde{\sigma}^{A_1}_{0|x_1} + \sum_{x_2} F^{(2)}_{x_2} \, \tilde{\sigma}^{A_2}_{0|x_2} + \sum_{x_1,x_2} F_{x_1x_2} \, \tilde{\sigma}_{00|x_1x_2}}{} \leq -0.0258\,,
\end{align}
where $\tilde{\sigma}^{A_1}_{0|x_1}  =  \sum_{a_2} \tilde{\sigma}_{0a_2|x_1x_2}$, $\tilde{\sigma}^{A_2}_{0|x_2}  =  \sum_{a_1} \tilde{\sigma}_{a_10|x_1x_2}$ and $\tilde{\rho}_B = \sum_{a_1a_2} \tilde{\sigma}_{a_1a_2|x_1x_2}$ are well-defined marginal assemblages whenever $\{\tilde{\sigma}_{a_1a_2|x_1x_2}\}$ is no-signalling. 
Table \ref{table:ineq} presents the explicit form of the operators $\{F_R, F^{(1)}_{x_1}, F^{(2)}_{x_2}, F_{x_1x_2}\}$.

\begin{table}
\begin{center} $ Fr = \left( \begin{array}{cccc}-2.65 &-0.495+1.74i &0.477+2.54i &0 \\ 
-0.495-1.74i &-4.33 &0 &0.477+2.54i \\ 
0.477-2.54i &0 &-4.33 &0.495-1.74i \\ 
0 &0.477-2.54i &0.495+1.74i &-2.65 
\end{array} \right)$
\\
$F^{(1)}_{0} = \left( \begin{array}{cccc}-1.17 &1.19+0.488i &-1.23+1.35i &0 \\ 
1.19-0.488i &-2.13 &0 &-1.23+1.35i \\ 
-1.23-1.35i &0 &-2.13 &-1.19-0.488i \\ 
0 &-1.23-1.35i &-1.19+0.488i &-1.17
\end{array} \right)$ 
 \\ 
$F^{(2)}_{0} = \left( \begin{array}{cccc}-1.16 &1.24-1.09i &-0.761+2.09i &0 \\ 
1.24+1.09i &-2.96 &0 &-0.761+2.09i \\ 
-0.761-2.09i &0 &-2.96 &-1.24+1.09i \\ 
0 &-0.761-2.09i &-1.24-1.09i &-1.16
\end{array} \right)$ 
 \\$F^{(1)}_{1} = \left( \begin{array}{cccc}-1.51 &0.134-1.58i &-0.0469+1.96i &0 \\ 
0.134+1.58i &-1.01 &0 &-0.0469+1.96i \\ 
-0.0469-1.96i &0 &-1.01 &-0.134+1.58i \\ 
0 &-0.0469-1.96i &-0.134-1.58i &-1.51
\end{array} \right)$ 
 \\ 
$F^{(2)}_{1} = \left( \begin{array}{cccc}-1.4 &-0.0294-0.939i &1.07+1.09i &0 \\ 
-0.0294+0.939i &-0.0605 &0 &1.07+1.09i \\ 
1.07-1.09i &0 &-0.0605 &0.0294+0.939i \\ 
0 &1.07-1.09i &0.0294-0.939i &-1.4
\end{array} \right)$
 \\
$F_{0 0} = \left( \begin{array}{cccc}-1.11 &-1.58-0.678i &-0.746-0.752i &0  \\ 
-1.58+0.678i &3.52 &0 &-0.746-0.752i  \\ 
-0.746+0.752i &0 &3.52 &1.58+0.678i  \\ 
0 &-0.746+0.752i &1.58-0.678i &-1.11   
\end{array} \right)$ 
 \\ 
$F_{0 1} = \left( \begin{array}{cccc}0.265 &-1.44-0.441i &-0.125+1.47i &0  \\ 
-1.44+0.441i &-4.12 &0 &-0.125+1.47i  \\ 
-0.125-1.47i &0 &-4.12 &1.44+0.441i  \\ 
0 &-0.125-1.47i &1.44-0.441i &0.265   
\end{array} \right)$ 
\\ 
$F_{1 0} = \left( \begin{array}{cccc}1.1 &-1.35-0.375i &-0.0531-1.55i &0  \\ 
-1.35+0.375i &0.17 &0 &-0.0531-1.55i  \\ 
-0.0531+1.55i &0 &0.17 &1.35+0.375i  \\ 
0 &-0.0531+1.55i &1.35-0.375i &1.1   
\end{array} \right)$ 
 \\ 
$F_{1 1} = \left( \begin{array}{cccc}0.519 &1.22-0.952i &0.14-1.11i &0  \\ 
1.22+0.952i &-0.213 &0 &0.14-1.11i  \\ 
0.14+1.11i &0 &-0.213 &-1.22+0.952i  \\ 
0 &0.14+1.11i &-1.22-0.952i &0.519   
\end{array} \right)$ 
 \end{center} 
\caption{The operators that define the steering inequality \eqref{eq:steeringineq}. The maximum value of the steering functional (i.e. the lhs of \eqref{eq:steeringineq}) by almost-quantum assemblages is $-0.0258$. Any assemblage that yields a value larger than that has therefore no quantum realisation. } 
\label{table:ineq}
 \end{table}

Every almost-quantum assemblage $\{\tilde{\sigma}_{a_1a_2|x_1x_2}\}$ (and therefore every quantum one) satisfies Eq.~\eqref{eq:steeringineq}. However, our particular assemblage $\{\sigma_{a_1a_2|x_1x_2}\}$ yields a value of $0.2044$ for the steering functional (i.e. the lhs of \eqref{eq:steeringineq}), hence violating Eq.~\eqref{eq:steeringineq}. This certifies the post-quantumness of $\{\sigma_{a_1a_2|x_1x_2}\}$. 

\end{document}